\documentclass[11pt]{article}
\usepackage[utf8]{inputenc}
\usepackage[english]{babel}
\usepackage{comment}
\usepackage{tikz}
\usepackage{tikz-3dplot}
\usepackage{amsmath, amssymb}
\usepackage{ifthen}
\usepackage{makeidx}
\usepackage{mathtools}
\usepackage{algorithm}
\usepackage{algorithmicx}
\usepackage{algpseudocode}
\usetikzlibrary {math,3d,calc}
\makeindex
\graphicspath{{./pics/}}
\usepackage{tikz}
\usepackage{bm}
\usetikzlibrary{arrows.meta}
\usetikzlibrary{shapes.multipart, positioning}
\usetikzlibrary{decorations.pathmorphing, decorations.markings}
\usepackage{multicol}
\usepackage{amsfonts,amsthm,amsmath,amssymb,enumitem}
 \usepackage{caption}
\usepackage{graphicx,float,subcaption}
\usepackage{verbatim}
 
\usepackage{amsthm}
\usepackage{graphicx}
\usepackage{caption}
\usepackage{booktabs}
\usepackage{hyperref}
\usepackage{threeparttable}

\usepackage{tikz}
\usetikzlibrary{arrows.meta,calc,positioning,fit,backgrounds}

\newtheorem{theorem}{Theorem}[section]
\newtheorem{lemma}{Lemma}[section]
\newtheorem{proposition}[theorem]{Proposition}

\newtheorem{remark}{Remark}[section]

\newtheorem{corollary}[theorem]{Corollary}

\numberwithin{equation}{section}
\allowdisplaybreaks
\date{}

\begin{document}

\title{On the uniqueness of the discrete Calder\'{o}n problem on multi-dimensional lattices\thanks{The work of M. Deng is supported by Hong Kong Research Grants Council via the Hong Kong PhD Fellowship Scheme. The work of B. Jin is supported by Hong Kong RGC General Research Fund (Projects 14306423 and 14306824),
ANR / RGC Joint Research Scheme (A-CUHK402/24) and a start-up fund from The Chinese University of Hong
Kong.}}
\author{
Maolin Deng\thanks{Department of Mathematics, The Chinese University of Hong Kong, Shatin, N.T., Hong Kong (\texttt{mldeng@link.cuhk.edu.hk}, \texttt{b.jin@cuhk.edu.hk})}
\and Bangti Jin\footnotemark[2]
}
 
\maketitle

\begin{abstract}
In this work, we investigate the discrete Calder\'{o}n problem on (hyper)cubic grid graphs in dimension three or higher. The discrete Calder\'{o}n problem is concerned with determining whether the discrete Dirichlet-to-Neumann matrix, which links boundary potentials to boundary current responses, can uniquely determine the conductivity values on the graph edges. We provide an affirmative answer to the question, thereby extending the classical uniqueness result of Curtis and Morrow for two-dimensional square lattices. The proof employs a slicing technique that decomposes the problem into lower-dimensional components. Additionally, we support the theoretical finding with numerical experiments that illustrate the effectiveness of the approach.
\end{abstract}

%\tableofcontents
%\printindex
\section{Introduction}
The Calder\'{o}n problem named after Alberto Calder\'{o}n falls into the class of inverse boundary value problems, mostly investigated in an open  bounded domain $\Omega\subseteq\mathbb R^d$, and focuses on the recovery of the interior conductivity only by the measurements on the boundary $\partial\Omega$. Physically, the interior electric potential satisfies Kirchhoff's current law, which takes the form of a second-order elliptic partial differential equation in the divergence form. Since the seminal work of Alberto Calder\'{o}n \cite{Calderon:1980}, the Calder\'{o}n problem has been extensively studied in the last forty years, with many deep theoretical results \cite{Borcea,Uhlmann:2009,FeldmanSaloUhlmann:2025}. 

The discrete Calder\'{o}n problem focuses on a finite graph $G=(E,D,\partial D)$. The set $E$ contains edges that connect two nodes in the set $D\cup\partial D$. Roughly speaking, the electric system on the graph $G$ can be viewed as connecting homogeneous electric wires according to the graph structure $G$. Then the electric potential on the edges and nodes, that is, the wires and their connecting points, is fully determined by its value on nodes $D\cup \partial D$, and by a linear interpolation of nodal values for each edge.

These two versions of the Calder\'{o}n problem (continuous versus discrete) represent two fundamentally different ways to formalize the physical problem. The mathematical tools utilized to investigate these two  versions differ significantly. In particular, very different excitations (i.e., the boundary potential chosen to induce the current measurement) are employed in the analysis. For the continuous Calder\'{o}n problem, the so-called complex geometric optics solutions \cite{Uhlmann1987} with high-frequency oscillations represent one powerful tool to derive theoretical results. In contrast, in the discrete case, the high-frequency property is inherently incompatible with the finite graph structure, and a new type of excitation has been devised, i.e., excitations that are able to localize all interior potential within a sub-domain of the graph \cite{CurtisMorrow:1990}. 

Let $G$ be a graph with vertex set $\overline D = D \cup \partial D$, where $D$ and $\partial D$ are disjoint finite sets. Its edge set is $E = \{ pq : pq \in \overline D\}$, the unordered pairs on $\overline D$. For each vertex $p \in \overline D$, the neighbor set is $\mathcal{N}(p) = \{ q \in \overline D : pq \in E \}$. Every edge $pq$ is assigned a positive conductivity $\gamma_{pq} = \gamma_{qp} > 0$. 
Real-valued functions on any set $V$ are identified with vectors in $\mathbb{R}^V$, and we denote the potential function by $\mathbf{u} \in \mathbb{R}^{\overline D}$.

Physically, the potential $\mathbf{u}$ satisfies Kirchhoff's current law on the interior nodes, and mathematically, we express it via the discrete Laplacian operator $\Delta_\gamma$, defined by
\[
(\Delta_\gamma \mathbf{u})_p = \sum_{q \in \mathcal{N}(p)} \gamma_{qp} (\mathbf{u}_q - \mathbf{u}_p), \quad p \in D.
\]
One can thus formulate the Dirichlet problem:
\begin{equation}
    \left\{\begin{aligned}
        (\Delta_\gamma \mathbf{u})_p &= 0,\quad p \in D, \\
        \mathbf{u}_p &= \varphi_p,\quad p \in \partial D,
    \end{aligned}\right.\label{eqn:model}
\end{equation}
where $\varphi\in \mathbb{R}^{\partial D}$ is the Dirichlet boundary condition. The unique solvability of problem \eqref{eqn:model} is well known for general graphs \cite{CurtisMorrow:1990}. Then one can define an injective solution operator $S_\gamma$ by
\begin{equation}
    \begin{aligned}
    S_\gamma : \mathbb{R}^{\partial D} &\to \mathbb{R}^{D \cup \partial D}, \quad
    \varphi \mapsto \mathbf{u},
    \end{aligned}\label{eqn:Solution Map}
\end{equation}
where \(\mathbf{u}\) is the unique solution of problem \eqref{eqn:model} with the Dirichlet condition \(\varphi\).
Assume that each boundary vertex \(p \in \partial D\)  is connected to precisely one interior vertex, that is, \(\mathcal{N}(p) = \{q\} \subseteq D\). Then the boundary current data can be represented as
\[
\psi_p := \sum_{s \in \mathcal{N}(p)} \gamma_{sp} (\mathbf{u}_s - \mathbf{u}_p) = \gamma_{pq} (\mathbf{u}_q - \mathbf{u}_p).
\]
Formally, one can also define
\begin{equation}
    \begin{aligned}
    D_\gamma : \mathbb{R}^{D \cup \partial D} &\to \mathbb{R}^{\partial D}, \quad
    \mathbf{u} \mapsto \psi.
\end{aligned}\label{def:Neumann Map}
\end{equation}
The Dirichlet-to-Neumann (DtN)  matrix \(\Lambda_\gamma:\mathbb{R}^{\partial D}\to \mathbb{R}\) is the composition of these two operators:
\[
\Lambda_\gamma = D_\gamma \circ S_\gamma.
\]

Then the discrete Calder\'{o}n problem reads: Given the DtN matrix \(\Lambda_\gamma\), is it possible to uniquely identify the conductivity \(\gamma\)? In the literature it is also known as the discrete inverse conductivity problem. In their seminal work \cite{CurtisMorrow:1990}, Curtis and Morrow established the uniqueness of the discrete Calder\'{o}n problem on two-dimensional square lattices.  This study aims to extend the result of Curtis and Morrow \cite{CurtisMorrow:1990} to three and higher dimensional lattice structures.

Now we precisely define the graph \(G = (E, D, \partial D)\) in a high-dimensional grid. The set $D$ of interior nodes consists of all integer tuples \((x_i)_{i=1}^d\in \mathbb Z^d\) such that each coordinate $x_i$ adheres to the constraints \(1 \leq x_i \leq n\). This can be expressed as 
\begin{equation}
    D \coloneqq \{ (x_i)_{i=1}^d \mid x_i \in \mathbb{Z},\ 1 \leq x_i \leq n \}.
    \label{eqn:D_RECT}
\end{equation}
The set \(\partial D\) of  boundary vertices is comprised of nodes for which exactly one coordinate is equal to either \(0\) or \(n+1\), while the remaining coordinates range over the set \(\{1, \ldots, n\}\):
\begin{equation}
    \partial D \coloneqq \{p\in\mathbb Z^d:d(p,D):=\min_{q\in D}\|q-p\|_{\ell^1} = 1\},
    \label{eqn:Part_D_RECT}
\end{equation}
where $\|\cdot\|_{\ell^1}$ denotes the $\ell^1$ norm of vectors.
The edge set \(E\) consists of unordered pairs \(\{p, q\}\), where \(p = (p_i)_{i=1}^d\) and \(q = (q_i )_{i=1}^d\) are vertices in \(D \cup \partial D\) that meet the following two criteria: (i) $\|p-q\|_{\ell^1}=1$; and (ii) The pair does not reside entirely within the boundary, i.e., \(\{p, q\} \not\subseteq \partial D\). This defines a multi-dimensional hypercubic lattice graph:
\begin{equation}
    \label{def:3dGrids}
    G = (E, D, \partial D).
\end{equation}

The main result of this work is the following uniqueness result for the discrete Calder\'{o}n problem on multi-dimensional hypercubic graphs. 
\begin{theorem}[Uniqueness]
For the multi-dimensional grid graph \(G\) given in \eqref{def:3dGrids}, the DtN matrix \(\Lambda_\gamma\) uniquely determines the conductivity \(\gamma\).
\label{thm:main}
\end{theorem}

Theorem \ref{thm:main} extends the result of Curtis and Morrow  \cite{CurtisMorrow:1990} on the uniqueness of the discrete Calder\'{o}n problem for two-dimensional square lattices to multi-dimensional hypercubic graphs.
The proof of Theorem \ref{thm:main} employs mathematical induction and a slicing technique that effectively breaks down the complex lattice into more manageable components. The uniqueness of the conductivity \(\gamma\) is established through an inductive procedure that is applied slice by slice, based on the slice decomposition strategy. The inductive step is based on refined characterizations of the algebraic properties of the solution spaces $\mathcal{U}^{(t)}$ on hypercubic lattices. Moreover, the analysis strategy is constructive, and lends itself to an efficient algebraic reconstruction algorithm that can be used to recover the conductivity $\gamma$ from the DtN matrix $\Lambda_\gamma$. In Section \ref{Sec:Numerical}, we present the reconstruction algorithm and numerical illustrations to show the feasibility of the reconstruction.

Finally, we discuss related works on the discrete Calder\'{o}n problem. The problem in the two-dimensional case has been extensively studied since the seminal work of Curtis and Morrow \cite{CurtisMorrow:1990}. There are many interesting theoretical results for the discrete Calder\'{o}n problem on planar graphs, including square lattices  \cite{CurtisMorrow:1991} and circular planar graphs \cite{CurtisMorrow:1994,deVerdiGitler:1996,CurtisIngermanMorrow:1998,Ingerman:2000}. The work \cite{ChungBerenstein:2005} proved a uniqueness result from one boundary measurement under a certain monotonicity assumption. Boyer et al \cite{BoyerGarzella:2016} investigated the recovery of a complex-valued conductivity (or complex-valued potential) from the DtN matrix $\Lambda_\gamma$, and established the unique recovery up to a set of measure zero. { From the DtN matrix on a tree graph, Gernandt and Rohleder \cite{GernandtRohleder:2022} proved an explicit formula which relates the DtN matrix to the pairwise weighted distances of the leaves of the tree (and thus the weighted tree).} 

The discrete Calder\'{o}n problem in the multi-dimensional / non-planar case has only been scarcely studied so far. On cylindrical networks, Lam and Pylyavskyy  \cite{LamPylyavskyy:2012} investigated the recovery of the conductivity from the so-called $R$-response matrix,  proved that the problem generally does not have uniqueness, and established uniqueness for a certain class of purely cylindrical networks. The analysis in \cite{LamPylyavskyy:2012} relies heavily on grove combinatorics. Ara\'{u}z et al \cite{Arauz:2016} treated cubic lattices for a particular excitation, and proved the unique recovery of the conductivity from the Robin-to-Neumann map. The proof relies on proving that certain overdetermined boundary value problems on graphs admit a unique solution. Compared with the work \cite{Arauz:2016}, we employ the space of admissible excitations as the main tool, and the approach extends readily to general multi-dimensional hypercubic lattice graphs and possibly graphs of more general shapes. Moreover, the  proposed approach is constructive in the sense that it lends itself directly to an algebraic reconstruction algorithm similar to the Curtis-Morrow algorithm for square lattices \cite{CurtisMorrow:1991}.

The works \cite{Morioka:2011,ChungGilbert:2017,HorvathMarko:2016,Oberlin:2010} investigated a closely related inverse problem of recovering the potential in the discrete Schr\"{o}dinger operator from the (partial) DtN type matrix. Note that due to the lack of a neat Liouville type transformation on graphs, the discrete inverse Schr\"{o}dinger problem is not equivalent to the discrete Calder\'{o}n problem.  Interestingly, in the work \cite{HorvathMarko:2016}, the slicing decomposition was also crucially employed in the uniqueness proof. Although this work shares the spirit with \cite{HorvathMarko:2016}, the recovery of the conductivity $\gamma$ requires several additional investigations of the algebraic properties of the solution spaces $\mathcal{U}^{(t)}$, and thus the analysis is much more intricate.
Also we refer interested readers to \cite{BlastenIsozaki:2023,BlastenIsozakiLassasLu:2023,Corbett:2025} for other recent contributions on inverse problems on graphs. 

{ Formally, the discrete Calder\'{o}n problem can be viewed as a discretization of the continuous version. By drawing on this connection and suitable discretization schemes, Borcea and her collaborators \cite{BorceaDruskin:2008,BorceaDruskinMamonov:2013,BorceaDruskinMamonov:2010a,BorceaDruskinMamonov:2010,BorceaMamonov:2017} have developed several effective numerical schemes for the continuous Calder\'{o}n problem. Nonetheless, to the best of our knowledge, rigorously establishing the connection of the uniqueness and stability results between these two versions of the Calder\'{o}n problem is still missing.}

The rest of the paper is organized as follows.
In Section~\ref{Sec:Corner}, we construct a specifically designed boundary potential that confines nonzero currents to a  corner region. In Section~\ref{Sec:TOPOLOGY}, we  collect three fundamental topological properties of grid graphs, which are essential for the proof of Theorem~\ref{thm:main}. In  Section~\ref{Sec:PROOF}, we present the proof of Theorem \ref{thm:main}. The proof leads to an explicit reconstruction algorithm that incrementally recovers the conductivity \(\gamma\), slice by slice, relying solely on the DtN matrix \(\Lambda_\gamma\). Finally, in Section \ref{Sec:Numerical}, we present numerical results to illustrate the algorithm derived from the uniqueness proof. 

Throughout this work, we denote the space of real-valued functions defined on a set \(U\) by \(\mathbb{R}^U\). We extend functions defined on subsets \(U \subset V\) to functions on \(V\) by assuming they are zero outside \(U\) and then we can define the inclusion $\mathbb R^U\subseteq \mathbb R^V$. Given a matrix \(A:\mathbb R^{V_1}\rightarrow \mathbb R^{V_2}\), we define its sub-matrix with column indices \(U_1\subseteq V_1\) and row indices \(U_2\subseteq V_2\) by \(A(U_2;U_1)\). {Given a subset $U\subset \overline{D}$, the notation $\big|_{U}$ denotes the restriction of functions from the domain $\overline D$ to the region $U$, and the convention also applies to vectors defined on the edge set $E$. For any two subsets $V_1,V_2\subset \overline{D}$, we denote by $E(V_1,V_2)$ the subset of the set $E$ comprised of edges starting from the vertex set $V_1$ and ending at $V_2$.}

\section{Corner excitations}\label{Sec:Corner}

In this section, we describe a special type of boundary excitations that induce well localized potentials and give their important properties. This is the main technical tool to carry out the inductive steps along the slices.
The inductive step starts from a subdomain with the known conductivity, and then aims to recover the conductivity on a larger region. That is, the domain with known conductivity grows with the induction step, which eventually covers the entire edge set $E$. To rigorously formulate the inductive step, we define a class of growing regions slice by slice as
\begin{subequations}\label{eq:level_sets}
\begin{align}
L_t &= \left\{ (x_i)_{i=1}^d \in D : \sum_{i=1}^dx_i = t \right\},&L_t^{\mathcal{S}} &= \bigcup_{\ell=0}^t L_\ell = \left\{ (x_i)_{i=1}^d \in D : \sum_{i=1}^dx_i \le t \right\},  \\
K_t &= \left\{ (x_i)_{i=1}^d \in\partial D : \sum_{i=1}^dx_i = t \right\},&K_t^{\mathcal{S}} &= \bigcup_{\ell=0}^t K_\ell = \left\{ (x_i)_{i=1}^d \in\partial D : \sum_{i=1}^dx_i \le t \right\}.
\end{align}
\end{subequations}
The inductive step is that, given the conductivity $\gamma$ on the edge set $ $, we recover $\gamma$ on the enlarged edge set $E(L_{t}^\mathcal S,\overline{D})$. 
The analysis requires also considering a subset \(J_t^{\mathcal{S}}\) on the boundary $\partial D$ that lies in the lower part of the graph $G$ divided by $L_{t+1}$:
\begin{subequations}\label{eq:boundary_classification}
\begin{align}
K_t^+ &= \{ (x_i)_{i=1}^d  \in K_t : \max_i x_i = n+1 \}, &K_t^{\mathcal{S}+} &= \bigcup_{\ell=0}^t K_\ell^+,\\
K_t^- &= \{ (x_i)_{i=1}^d  \in K_t : \min_i x_i = 0 \},&K_t^{\mathcal{S}-}& = \bigcup_{\ell=0}^t K_\ell^-.\\
J_t &= K_t^- \cup K_{t+1}^+,  
&J_t^{\mathcal{S}}  &= K_t^{\mathcal{S}-} \cup K_{t+1}^{\mathcal{S}+}.\label{eq:J_sets}
\end{align}
\end{subequations}
We refer the reader to Fig. \ref{fig:interface_partition} for a schematic illustration of the interface partition in the three-dimensional case.

\begin{figure}[htb!]
    \centering
    \begin{tabular}{cc}
        \includegraphics[width=0.35\linewidth]{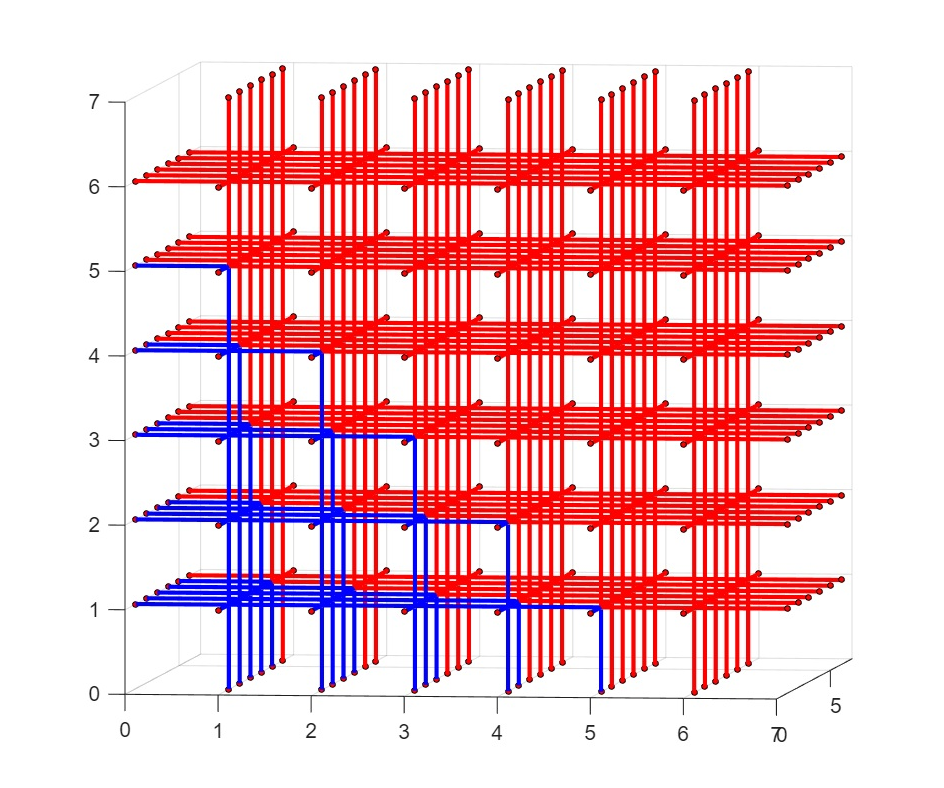} & \includegraphics[width=0.35\linewidth]{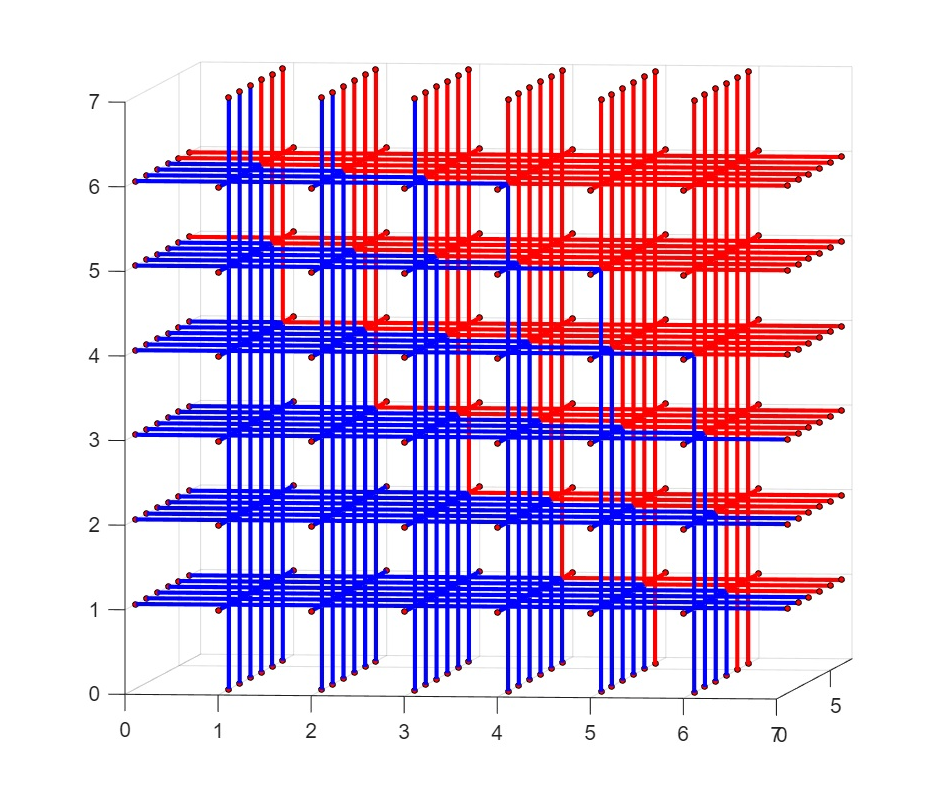} \\
        \(t \leq n+d-1\)& \(t > n+d-1\)\end{tabular}
    \caption{The interface partition within the domain \(D\subset \mathbb{Z}^3\). The blue region is in \(L_{t+1}^\mathcal{S}\).}
    \label{fig:interface_partition}
\end{figure}

Roughly speaking, to recover the conductivity $\gamma$ on the current sub-region $L_{t+1}^\mathcal S\cup K_{t+1}^\mathcal S$, we select a set of boundary excitations that ``focus" on the boundary of the sub-region $L_{t+1}^\mathcal{S}\cup K_{t+1}^\mathcal{S} $ and use its boundary measurements to finish the inductive step. The excitations are exactly in the kernel of the following map
\begin{equation}
\label{eq:Tt_operator}T^{(t)}  := \Lambda_\gamma(\partial D \setminus J_t^{\mathcal{S}}; J_t^{\mathcal{S}}).
 \end{equation}

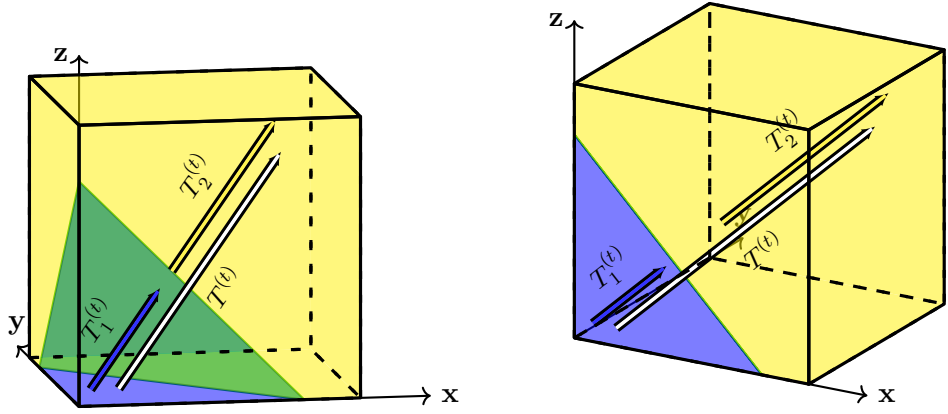
\begin{figure}[hbt!]
    \centering
  
\begin{tabular}{cc}
  \tdplotsetmaincoords{80}{-10} % 视角稍微降低，方位角稍微调整

\begin{tikzpicture}[scale=1.5,line join=round,tdplot_main_coords ]

  % 坐标轴，增加箭头和标签
  \draw[->, thick] (0,0,0) -- (2.5,0,0) node[right] {\(\mathbf{x}\)};
  \draw[->, thick] (0,0,0) -- (0,2.5,0) node[above] {\(\mathbf{y}\)};
  \draw[->, thick] (0,0,0) -- (0,0,2.5) node[left] {\(\mathbf{z}\)};

  % 切割平面交点
  \coordinate (A) at (1.6,0,0);
  \coordinate (B) at (0,1.6,0);
  \coordinate (C) at (0,0,1.6);
  % 切割平面三角形填充
  \fill[green!90,opacity=0.8] (A) -- (B) -- (C) -- cycle;

  % 切割平面边界线（绿色，稍加粗）
  \draw[green!70!black,thick] (A) -- (B) -- (C) -- cycle;
  % 各面着色（保持原色，调整透明度）
  \begin{scope}[opacity=0.3]
    % 前面 (z=2)
    \begin{scope}[canvas is xy plane at z=2]
      \clip (0,0) rectangle (2,2);
      \fill[yellow] (0,0) rectangle (2,2);
    \end{scope}

    % 右面 (x=2)
    \begin{scope}[canvas is yz plane at x=2]
      \clip (0,0) rectangle (2,2);
      \fill[yellow] (0,0) rectangle (2,2);
    \end{scope}

    % 顶面 (y=2)
    \begin{scope}[canvas is xz plane at y=2]
      \clip (0,0) rectangle (2,2);
      \fill[yellow] (0,0) rectangle (2,2);
    \end{scope}

    % 底面 (y=0)
    \begin{scope}[canvas is xz plane at y=0]
      \clip (0,0) rectangle (2,2);
      \fill[blue,domain=0:1] (0,0) -- (1.6,0) -- (0,1.6) -- cycle;
      \fill[yellow,domain=1:2] (1.6,0) -- (2,0) -- (2,2) -- (0,2) -- (0,1.6) -- cycle;
    \end{scope}

    % 左面 (x=0)
    \begin{scope}[canvas is yz plane at x=0]
      \clip (0,0) rectangle (2,2);
      \fill[blue,domain=0:1,opacity=0.3] (0,0) -- (1.6,0) -- (0,1.6) -- cycle;
      \fill[yellow,domain=1:2] (1.6,0) -- (2,0) -- (2,2) -- (0,2) -- (0,1.6) -- cycle;
    \end{scope}

    % 背面 (z=0)
    \begin{scope}[canvas is xy plane at z=0]
      \clip (0,0) rectangle (2,2);
      \fill[blue,domain=0:1,opacity=0.3] (0,0) -- (1.6,0) -- (0,1.6) -- cycle;
      \fill[yellow,domain=1:2] (1.6,0) -- (2,0) -- (2,2) -- (0,2) -- (0,1.6) -- cycle;
    \end{scope}
  \end{scope}

  % 向量箭头及标签（加粗箭头，标签字体稍大，位置微调）
  \draw[->, line width=1.3pt, double=white, double distance=1pt, 
    arrows={-Latex[length=5pt, width=5pt]}, 
    postaction={draw=blue!80, shorten >=1.5pt, -}] 
    (0.1,0.1,0.1) -- (0.7,0.7,0.7) 
    node[midway, sloped, above=2pt, scale=0.9] {\textit{\( T_1^{(t)} \)}};

  \draw[->, line width=1.3pt, double=white, double distance=1pt, 
    arrows={-Latex[length=5pt, width=5pt]}, 
    postaction={draw=yellow!80, shorten >=1.5pt, -}] 
    (0.8,0.8,0.8) -- (1.7,1.7,1.7)
    node[midway, sloped, above=2pt, scale=0.9] {\textit{\( T_2^{(t)} \)}};

  \draw[->, line width=1.3pt, double=white, double distance=1pt, 
    arrows={-Latex[length=5pt, width=5pt]}, 
    postaction={draw=white!80, shorten >=1.5pt, -}] 
    (0.3,0.1,0.1) -- (1.7,1.5,1.5)
    node[midway, sloped, below=2pt, scale=0.9] {\textit{\( T^{(t)} \)}};

  % 立方体各面边框，虚线和实线区分，线宽加粗
  \begin{scope}[canvas is xy plane at z=0]
    \draw[loosely dashed, line width=1.2pt] (0,0) rectangle (2,2);
  \end{scope}
  \begin{scope}[canvas is yz plane at x=2]
    \draw[loosely dashed, line width=1.2pt] (0,0) rectangle (2,2);
  \end{scope}
  \begin{scope}[canvas is xz plane at y=2]
    \draw[loosely dashed, line width=1.2pt] (0,0) rectangle (2,2);
  \end{scope}
  \begin{scope}[canvas is xy plane at z=2]
    \draw[line width=1.2pt] (0,0) rectangle (2,2);
  \end{scope}
  \begin{scope}[canvas is yz plane at x=0]
    \draw[line width=1.2pt] (0,0) rectangle (2,2);
  \end{scope}
  \begin{scope}[canvas is xz plane at y=0]
    \draw[line width=1.2pt] (0,0) rectangle (2,2);
  \end{scope}

\end{tikzpicture} &  \tdplotsetmaincoords{70}{30} % 视角稍微降低，方位角稍微调整

\begin{tikzpicture}[scale=1.4,line join=round,tdplot_main_coords ]

  % 坐标轴，增加箭头和标签
  \draw[->, thick] (0,0,0) -- (2.5,0,0) node[right] {\(\mathbf{x}\)};
  \draw[->, thick] (0,0,0) -- (0,2.5,0) node[above] {\(\mathbf{y}\)};
  \draw[->, thick] (0,0,0) -- (0,0,2.5) node[left] {\(\mathbf{z}\)};

  % 切割平面交点
  \coordinate (A) at (1.6,0,0);
  \coordinate (B) at (0,1.6,0);
  \coordinate (C) at (0,0,1.6);
  % 切割平面三角形填充
  \fill[green!90,opacity=0.8] (A) -- (B) -- (C) -- cycle;

  % 切割平面边界线（绿色，稍加粗）
  \draw[green!70!black,thick] (A) -- (B) -- (C) -- cycle;
  % 各面着色（保持原色，调整透明度）
  \begin{scope}[opacity=0.3]
    % 前面 (z=2)
    \begin{scope}[canvas is xy plane at z=2]
      \clip (0,0) rectangle (2,2);
      \fill[yellow] (0,0) rectangle (2,2);
    \end{scope}

    % 右面 (x=2)
    \begin{scope}[canvas is yz plane at x=2]
      \clip (0,0) rectangle (2,2);
      \fill[yellow] (0,0) rectangle (2,2);
    \end{scope}

    % 顶面 (y=2)
    \begin{scope}[canvas is xz plane at y=2]
      \clip (0,0) rectangle (2,2);
      \fill[yellow] (0,0) rectangle (2,2);
    \end{scope}

    % 底面 (y=0)
    \begin{scope}[canvas is xz plane at y=0]
      \clip (0,0) rectangle (2,2);
      \fill[blue,domain=0:1] (0,0) -- (1.6,0) -- (0,1.6) -- cycle;
      \fill[yellow,domain=1:2] (1.6,0) -- (2,0) -- (2,2) -- (0,2) -- (0,1.6) -- cycle;
    \end{scope}

    % 左面 (x=0)
    \begin{scope}[canvas is yz plane at x=0]
      \clip (0,0) rectangle (2,2);
      \fill[blue,domain=0:1,opacity=0.3] (0,0) -- (1.6,0) -- (0,1.6) -- cycle;
      \fill[yellow,domain=1:2] (1.6,0) -- (2,0) -- (2,2) -- (0,2) -- (0,1.6) -- cycle;
    \end{scope}

    % 背面 (z=0)
    \begin{scope}[canvas is xy plane at z=0]
      \clip (0,0) rectangle (2,2);
      \fill[blue,domain=0:1,opacity=0.3] (0,0) -- (1.6,0) -- (0,1.6) -- cycle;
      \fill[yellow,domain=1:2] (1.6,0) -- (2,0) -- (2,2) -- (0,2) -- (0,1.6) -- cycle;
    \end{scope}
  \end{scope}

  % 向量箭头及标签（加粗箭头，标签字体稍大，位置微调）
  \draw[->, line width=1.3pt, double=white, double distance=1pt, 
    arrows={-Latex[length=5pt, width=5pt]}, 
    postaction={draw=blue!80, shorten >=1.5pt, -}] 
    (0.1,0.1,0.1) -- (0.5,0.5,0.5) 
    node[midway, sloped, above=2pt, scale=0.9] {\textit{\( T_1^{(t)} \)}};

  \draw[->, line width=1.3pt, double=white, double distance=1pt, 
    arrows={-Latex[length=5pt, width=5pt]}, 
    postaction={draw=yellow!80, shorten >=1.5pt, -}] 
    (0.8,0.8,0.8) -- (1.7,1.7,1.7)
    node[midway, sloped, above=2pt, scale=0.9] {\textit{\( T_2^{(t)} \)}};

  \draw[->, line width=1.3pt, double=white, double distance=1pt, 
    arrows={-Latex[length=5pt, width=5pt]}, 
    postaction={draw=white!80, shorten >=1.5pt, -}] 
    (0.3,0.1,0.1) -- (1.7,1.5,1.5)
    node[midway, sloped, below=2pt, scale=0.9] {\textit{\( T^{(t)} \)}};

  % 立方体各面边框，虚线和实线区分，线宽加粗
  \begin{scope}[canvas is xy plane at z=0]
    \draw[loosely dashed, line width=1.2pt] (0,0) rectangle (2,2);
  \end{scope}
  \begin{scope}[canvas is yz plane at x=0]
    \draw[loosely dashed, line width=1.2pt] (0,0) rectangle (2,2);
  \end{scope}
  \begin{scope}[canvas is xz plane at y=2]
    \draw[loosely dashed, line width=1.2pt] (0,0) rectangle (2,2);
  \end{scope}
  \begin{scope}[canvas is xy plane at z=2]
    \draw[line width=1.2pt] (0,0) rectangle (2,2);
  \end{scope}
  \begin{scope}[canvas is yz plane at x=2]
    \draw[line width=1.2pt] (0,0) rectangle (2,2);
  \end{scope}
  \begin{scope}[canvas is xz plane at y=0]
    \draw[line width=1.2pt] (0,0) rectangle (2,2);
  \end{scope}

\end{tikzpicture}  
\end{tabular}
 
    \caption{$J_t^\mathcal 
    S$ {\rm(}blue{\rm)}, $L_{t+1}$ {\rm(}green{\rm)} and $\partial D\setminus J_t^\mathcal S$ {\rm(}yellow{\rm)} in the three dimensional case and its projected view.}
    \label{fig:ThreeLine}
\end{figure}

We can factorize the map $T^{(t)}$ into two parts: one is the map from the potential supported on the boundary $J_t^{\cal S}$ to the potential supported in the region $L_{t+1}$, and the other is the map from the potential supported on $L_{t+1}$ to the boundary currents on $\partial D\setminus J_t^{\cal S}$. We refer the reader to Fig. \ref{fig:ThreeLine} for a schematic illustration in the three-dimensional case:  the sought-for excitations are precisely supported on the blue nodes with the boundary current vanishing at yellow nodes. More precisely, we define the map $T_1^{(t)}:\mathbb{R}^{J_t^\mathcal{S}}\to \mathbb{R}^{L_{t+1}}$ by \begin{align}
     T_1^{(t)} := S_\gamma(L_{t+1}; J_t^{\mathcal{S}}) .
    \end{align}
Then we denote the upper subgraph \(G^{(t)}\) (i.e., above the interface $L_{t+1}$) by
\begin{equation}\label{eq:subgraph_Gt}
G^{(t)} = \bigl(E^{(t)},\, D \setminus L_{t+1}^\mathcal{S},\,L_{t+1} \cup (\partial D\setminus J_t^{\mathcal{S}})\bigr),
\end{equation}
where the set \(E^{(t)}\) contains all edges in the set \(E\) start or end at one node in \(D \setminus L_{t}^\mathcal{S}\). In view of \cite[Remark 3]{CurtisMorrow:1990}, the unique solvability of the discrete Dirichlet problem on the subgraph $G^{(t)}$ holds,  and we denote by \( S_\gamma^{(t)} \) the solution operator on the subgraph \( G^{(t)} \). Hence the DtN map \(\Lambda^{(t)}\) on the subgraph $G^{(t)}$ is well defined. We can now describe the map $T_2^{(t)}: \mathbb{R}^{L_{t+1}} \to \mathbb{R}^{\partial D \setminus J_t^{\mathcal{S}}}$, from the potential at $L_{t+1}$ to the current at $\partial D\setminus J_t^{\cal S}$, by a submatrix of \(\Lambda^{(t)}\):
\begin{equation}\label{eq:T2t_operator}
T_2^{(t)} := \Lambda^{(t)}(\partial D \setminus J_t^{\mathcal{S}}; L_{t+1}) .
\end{equation}

The next result rigorously establishes the factorization.
\begin{proposition}\label{prop:factorization}
For any \( t \) in the range \( d-1 \leq t \leq dn - 1 \), the operator \( T^{(t)} \) can be factorized as 
\begin{equation*}
    T^{(t)} = T_2^{(t)} \circ T_1^{(t)}. 
\end{equation*}
\end{proposition}

\begin{proof}
Fix any arbitrary \(\varphi \in \mathbb{R}^{J_t^{\mathcal{S}}}\), and consider the zero extension of $\varphi$ to $\partial D$, still denoted by $\varphi$. The corresponding interior potential $\mathbf{u}$ is given by
$ S_\gamma \varphi$. By its definition, \(\mathbf{u}\) is \(\gamma\)-harmonic on \( G \), and also within the subgraph \( G^{(t)} \). By the definition of the operator \( T_1^{(t)}\), we have
\[
\mathbf{u}|_{L_{t+1}} = T_1^{(t)} \varphi.
\]
Since the discrete Dirichlet problem on \( G^{(t)} \) has a unique solution, the restriction of \(\mathbf{u}\) to the vertices of \( G^{(t)} \) (i.e., $\overline{D}\setminus J_t^\mathcal{S})$  can be represented by
\[
\mathbf{u}|_{\overline{D}\setminus J_t^\mathcal{S}} = S_\gamma^{(t)} \bigl( T_1^{(t)} \varphi \bigr). % \quad \text{in } { V(G^{(t)})}.
\]
Now we compute the boundary current and arrive at
\[
T^{(t)} \varphi = T_2^{(t)} \bigl( T_1^{(t)} \varphi \bigr).
\]
Since this relationship holds for any choice of \(\varphi\in \mathbb{R}^{J_t^{\mathcal{S}}}\), the desired factorization follows. 
\end{proof}

\begin{remark}
The factorization in Proposition \ref{prop:factorization} leads directly to the kernel inclusion relation:
\begin{equation}\label{eq:kernel_inclusion}
\ker T_1^{(t)} \subseteq \ker T^{(t)}.
\end{equation}
The identity \(\ker T_1^{(t)} = \ker T^{(t)}\) then follows from the injectivity of $T_2^{(t)}$, which will be proved in Section \ref{Sec:TOPOLOGY}.
\end{remark}

If the boundary potential $\varphi$ is chosen from $\ker T_1^{(t)}$, the interior potentials are  localized within the set $J_t^\mathcal{S} \cup L_t^\mathcal{S}$, since the Dirichlet boundary condition of the upper graph $G^{(t)}$ vanishes identically by the definition of the map $T_1^{(t)}$. We now define the solution spaces
\begin{equation}\label{eq:solution_space}
\mathcal{U}^{(t)} := \bigl\{\, S_\gamma \varphi \mid \varphi \in \ker T_1^{(t)} \bigr\} \big|_{J_t^\mathcal{S} \cup L_t^\mathcal{S}},
\end{equation}
where $S_\gamma$ is the solution operator defined in \eqref{eqn:Solution Map}. %and the notation $\big|_{J_t^\mathcal{S} \cup L_t^\mathcal{S}}$ denotes the restriction of functions from the domain $\overline D$ to the region $J_t^\mathcal{S} \cup L_t^\mathcal{S}$, on which the potential is naturally supported. 
Note that, when recovering the conductivity on the edge set $E(L_{t} \cup K_t,L_{t+1} \cup K_{t+1})$, only the boundary potentials in the quotient space 
\[
\ker T_1^{(t)} \big/ \ker T_1^{(t-1)}
\]
provide active constraints, since for any $\varphi\in\ker T_1^{(t-1)}$ no current flows through $E(L_{t} \cup K_t,L_{t+1} \cup K_{t+1})$. The quotient makes sense in view of the following inclusion result in the sense $\mathbb R^{J_{t-1}^\mathcal{S}}\subseteq \mathbb R^{J_{t }^\mathcal{S}} $.

\begin{proposition}\label{prop:monotonicity}
For \(d \leq t \leq dn-1\):
\[
\ker T_1^{(t-1)} \subseteq \ker T_1^{(t)} \   
\]
\end{proposition}

\begin{proof}
For any \(\varphi \in \ker T_1^{(t-1)}\), let \(\mathbf{u} := S_\gamma \varphi\). By the construction,
\[
\mathbf{u} = 0 \quad \text{in } D \setminus L_{t-1} ^\mathcal S\supseteq L_{t+1},
\]
with \(\mathrm{supp}(\varphi) \subseteq J_{t-1}^\mathcal{S} \subseteq J_t^\mathcal{S}\). This implies
$\varphi \in \ker T_1^{(t)}$ and hence the desired assertion $\ \ker T_1^{(t-1)} \subseteq \ker T_1^{(t)}$ holds. 
\end{proof}

The monotonicity of \(\mathcal{U}^{(t)}\) follows directly from Proposition \ref{prop:monotonicity}:
\[
\mathcal{U}^{(t-1)} = \{S_\gamma \varphi : \varphi \in \ker T_1^{(t-1)}\} \subseteq \{S_\gamma \varphi : \varphi \in \ker T_1^{(t)}\} = \mathcal{U}^{(t)}.\]
The set of localized potentials  $\mathcal U^{(t)}$ is a linear subspace of $\mathbb R^{J_t^\mathcal{S} \cup L_t^\mathcal{S}}$. The next result gives its complement. Essentially it is a rewriting of the governing equation in \eqref{eqn:model} for the localized solution in $\mathcal U^{(t)}$. 
\begin{lemma}\label{lem:CHAR}
For each node \( p \in \overline D \), define the vector \( \mathbf{v}_p \in \mathbb{R}^{\overline{D}} \) componentwise by
\begin{equation}
\label{eqn:vp}
(\mathbf{v}_p)_q \coloneqq
\begin{cases}
\gamma_{pq}, & \text{if } q \in \mathcal{N}(p), \\
-\sum_{r \in \mathcal{N}(p)} \gamma_{pr}, & \text{if } q = p, \\
0, & \text{otherwise}.
\end{cases}
\end{equation}
 Then, for any integer \( t \) satisfying \( d-1 \leq t \leq dn-1 \), the following orthogonal decomposition holds:
\[
\mathcal{U}^{(t)} \oplus \operatorname{span}\left\{ \mathbf{v}_p\big|_{L_t^{\mathcal{S}} \cup J_t^{\mathcal{S}}} : p \in L_{t+1}^{\mathcal{S}} \right\} = \mathbb{R}^{L_t^{\mathcal{S}} \cup J_t^{\mathcal{S}}}.
\]
\end{lemma}

\begin{proof}
Fix any \( \mathbf{u} \in \mathcal{U}^{(t)} \) and extend it onto $D\cup\partial D$ by zero (still denoted by $\mathbf u$). By the definition of potential solutions, for \( p \in L_{t+1}^{\mathcal{S}} \), we have 
$$ \sum_{q \in \mathcal{N}(p)} \gamma_{pq} (\mathbf{u}_q - \mathbf{u}_p) = 0,$$ or equivalently by $\operatorname{supp}(\mathbf u)=  L_t^{\mathcal{S}} \cup J_t^{\mathcal{S}}$, we have $\mathbf{v}_p\big|_{L_t^{\mathcal{S}} \cup J_t^{\mathcal{S}} }\cdot \mathbf u = 0 $.
Thus we have
\[
\operatorname{span}\left\{ \mathbf{v}_p\big|_{L_t^{\mathcal{S}} \cup J_t^{\mathcal{S}}} : p \in L_{t+1}^{\mathcal{S}} \right\} \subseteq \left( \mathcal{U}^{(t)} \right)^\perp.
\]
Conversely, suppose that \( \mathbf{u} \in ( \operatorname{span}_{p \in L_{t+1}^{\mathcal{S}}} \{ \mathbf{v}_p\big|_{L_t^{\mathcal{S}} \cup J_t^{\mathcal{S}}} \} )^\perp \). Then, 
\[
\mathbf{v}_p\big|_{L_t^{\mathcal{S}} \cup J_t^{\mathcal{S}}} \cdot \mathbf{u} = 0,\quad \forall p \in L_{t+1}^{\mathcal{S}}.
\]
After extending $\mathbf u$ onto $ D\cup\partial D$ by zero, the above equation implies that  \( -\Delta_\gamma \mathbf{u}_p = 0 \) for all \( p \in D \). Therefore, we have
\[
\mathbf{u} = S_\gamma(\mathbf{u}|_{\partial D}) \big|_{L_t^{\mathcal{S}}\cup J_t^{\mathcal{S}}} \in \mathcal{U}^{(t)},
\]
and there holds
$\mathcal{U}^{(t)} \supseteq ( \operatorname{span}\{ \mathbf{v}_p\big|_{L_t^{\mathcal{S}} \cup J_t^{\mathcal{S}}} : p \in L_{t+1}^{\mathcal{S}}\} )^\perp$.
Combining these two inclusions yields the desired orthogonal decomposition.
\end{proof}

\begin{remark}
The definition of excitations and their properties do not rely on the explicit grid graph structure. Whenever one can divide the graph with the boundary nodes into two connected parts, one can define such localized spaces, and the respective properties remain valid. 
\end{remark}

\section{Properties of lattice graphs}\label{Sec:TOPOLOGY}
Note that the unique determination of the conductivity $\gamma$ from the DtN matrix $\Lambda_\gamma$ is not necessarily valid for a general graph structure. In fact, under the given definitions of $D$ and $\partial D$, if we choose the edge set to be complete in $D$ -- that is, for every different nodes $p$ and $q$ in $D$, there exists an edge $pq$ in $E$ -- then the number of undetermined conductivity parameters, scaling as $O(n^{2d})$, exceeds the size of DtN matrix $\Lambda_\gamma$, which grows as $O((2dn^{d-1})^2)$, for large $n$. This naturally precludes the unique determination for graph structures, due to excessive connectivity (i.e. too many edge connections). 

In this section, we identify several fundamental properties of lattice graphs that ensure the unique determination of the conductivity $\gamma$. For each integer \( t \), consider the subgraph
\[
G'^{(t)} = \bigl(E'^{(t)},\, L_{t-1}^\mathcal{S},\, J_{t-1}^\mathcal{S} \cup L_t \bigr),
\]
where the edge set \( E'^{(t)} \) consists of all edges connecting nodes in \( L_{t-1}^\mathcal{S} \) to nodes in $ D\cup \partial D$. Associated to the subgraph \( G'^{(t)} \), there exists a DtN map \(\Lambda_\gamma'\), acting on the boundary potentials supported on $J_{t-1}^\mathcal{S}\cup L_t$.
Then we define the operator
\begin{equation}
    T_2'^{(t)} := \Lambda_\gamma'(J_{t-1}^\mathcal{S}; L_t),\label{eq:def_T2_prime}
\end{equation}
which maps the potentials prescribed on \( L_t \) to the induced currents on \( J_{t-1}^\mathcal{S} \).

The next lemma provides the key tool, which allows controlling the potential in the interior domain $L_{t }^\mathcal S$, using the potential and current both taken on $J_{t-1}^\mathcal S$. {It can be viewed as the Cauchy problem for the discrete Laplacian $\Delta_\gamma$. Analogous to the Cauchy problem for the Laplacian, problem \eqref{eqn:Subgraph_MIX} is expected to be highly ill-conditioned.}
\begin{lemma}\label{Lem:Cond}
Let \( d \leq t \leq dn \) and suppose that \(\gamma > 0\) is a conductivity on the lattice subgraph \( G'^{(t)} \). The mixed boundary condition problem for $\mathbf u\in\mathbb R^{L_{t}^\mathcal S}$,
\begin{equation}\label{eqn:Subgraph_MIX}
\left\{\begin{aligned}
\Delta_\gamma \mathbf{u} &= \mathbf{0} \quad \text{in } L_{t-1}^\mathcal{S}, \\
\mathbf{u} &= \mathbf{0} \quad \text{on } J_{t-1}^\mathcal{S}, \\
D_\gamma \mathbf{u} &= \mathbf{0} \quad \text{on } J_{t-1}^\mathcal{S},
\end{aligned}\right.
\end{equation}
has only the trivial solution, i.e., \(\mathbf{u} = \mathbf{0}\) in $L_t^\mathcal{S}$.
\end{lemma}
\begin{proof}
We prove the assertion \(\mathbf{u} = \mathbf{0}\) in \(L_t^\mathcal{S}\) by mathematical induction. To this end, let
\[
H_l = (J_{t-1}^\mathcal{S} \cup L_t^\mathcal{S}) \cap \{ (x_i)_i : x_1 = l \}, \quad H_l^\mathcal{S} = \bigcup_{i=0}^l H_i.
\]
Let \(\mathbf{e}_i = (0, \ldots, 0, {1}, 0, \ldots, 0) \in \mathbb{R}^d\) be the $i$-th Cartesian coordinate vector. Note that for each \(p \in H_l \setminus \partial D\) with \(l \geq 1\), the node \(q = p - \mathbf{e}_1\) is connected to exactly one node that lies outside \(H_{l-1}^\mathcal{S}\), namely \(p\) itself.
The base case \(\mathbf{u}|_{H_0} = \mathbf{0}\) holds by assumption. Now assume that \(\mathbf{u}|_{H_{l-1}^{\mathcal S}} = \mathbf{0}\). For each \(p \in H_l \setminus \partial D\), let \(q = p - \mathbf{e}_1\). Then there holds
\[
0 = \sum_{s \in \mathcal{N}(q)} \gamma_{sq} (u_s - u_q) = \gamma_{pq} u_p,
\]
which directly implies \(u_p = 0\). Since \(u_p = 0\) for \(p \in H_l \cap \partial D\) by assumption, the induction step is complete.
Hence, \(\mathbf{u} = \mathbf{0}\) in \(L_t^\mathcal{S}\).
\end{proof}

The validity of Lemma \ref{Lem:Cond} relies crucially on the grid graph structure. However, the next lemma is independent of the graph structure. 
\begin{lemma}\label{lem:linear_independence}
For every integer \( t \) with \( d \leq t \leq dn \) and any strictly positive conductivity \(\gamma > 0\), the following three statements are equivalent:
\begin{enumerate}
    \item[{\rm(i)}] Lemma~\ref{Lem:Cond} holds.
    \item[{\rm(ii)}] The set of vectors \(\{\mathbf{v}_p|_{L_{t-1}^\mathcal{S} \cup J_{t-1}^\mathcal{S}} : p \in L_t^\mathcal{S}\}\) is linearly independent, where the vectors \(\mathbf{v}_p\) are defined in \eqref{eqn:vp}.
    \item[{\rm(iii)}] The operator \( T_2'^{(t)} \), defined in \eqref{eq:def_T2_prime}, is injective.
\end{enumerate}
\end{lemma}
\begin{proof}
\textbf{(i) \(\Leftrightarrow\) (ii):}  
The statement (ii) asserts that
\begin{equation}\label{eqn:c-v}
\sum_{p \in L_t^\mathcal{S}} c_p \mathbf{v}_p \big|_{L_{t-1}^\mathcal{S} \cup J_{t-1}^\mathcal{S}} = \mathbf{0} \quad \implies \quad \mathbf{c} = \mathbf{0}.
\end{equation}
By interpreting the vector \(\mathbf{c}\) as a potential supported on \(G'^{(t)}\) with the zero boundary condition \(\mathbf{c}|_{J_{t-1}^\mathcal{S}} = \mathbf{0}\), the relation \eqref{eqn:c-v} is equivalent to the uniqueness in the statement (i):
\begin{equation*}
\left\{\begin{aligned}
\Delta_\gamma \mathbf{c} &= \mathbf{0} \quad \text{in } L_{t-1}^\mathcal{S}, \\
\mathbf{c} &= \mathbf{0} \quad \text{on } J_{t-1}^\mathcal{S}, \\
D_\gamma \mathbf{c} &= \mathbf{0} \quad \text{on } J_{t-1}^\mathcal{S},
\end{aligned}\right.
\quad \implies \quad \mathbf{c} = \mathbf{0}.
\end{equation*}
\medskip
\noindent\textbf{(i) \(\Leftrightarrow\) (iii):}  
Now suppose that the statement (iii) holds.  Consider the problem
\begin{equation*}
\left\{\begin{aligned}
\Delta_\gamma \mathbf{u} &= \mathbf{0} \quad \text{in } L_{t-1}^\mathcal{S}, \\
\mathbf{u} &= \mathbf{0} \quad \text{on } J_{t-1}^\mathcal{S}, \\
D_\gamma \mathbf{u} &= \mathbf{0} \quad \text{on } J_{t-1}^\mathcal{S}.
\end{aligned}\right.
\end{equation*}
Since \(D_\gamma \mathbf{u} = \mathbf{0}\) on \(J_{t-1}^\mathcal{S}\), it follows from the injectivity of \( T_2'^{(t)} \) that \(\mathbf{u} = \mathbf{0}\) on \(L_t \). Hence, \(\mathbf{u}\) vanishes on the entire boundary of the subgraph \(G'^{(t)}\). By the uniqueness of the discrete Dirichlet problem  {on \(G'^{(t)}\)}, we conclude \(\mathbf{u} = \mathbf{0}\) on \(L_{t}^\mathcal{S}\).
Conversely, suppose that the statement (i) holds. If \(T_2'^{(t)} \mathbf{x} = \mathbf{0}\), then the solution \(\mathbf{u}\) to
\[
\left\{\begin{aligned}
\Delta_\gamma \mathbf{u} &= \mathbf{0} \quad\text{in } L_{t-1}^\mathcal{S}, \\
\mathbf{u} &= \mathbf{0} \quad \text{on } J_{t-1}^\mathcal{S}, \\
\mathbf{u} &= \mathbf{x} \quad \text{on } L_t ,
\end{aligned}\right.
\]
also satisfies
\[
\left\{\begin{aligned}
\Delta_\gamma \mathbf{u} &= \mathbf{0} \quad \text{in } L_{t-1}^\mathcal{S}, \\
\mathbf{u} &= \mathbf{0} \quad \text{on } J_{t-1}^\mathcal{S}, \\
D_\gamma \mathbf{u} &={T_2'^{(t)}\mathbf x} = \mathbf{0} \quad\text{on } J_{t-1}^\mathcal{S}.
\end{aligned}\right.
\]
By (i), this implies \(\mathbf{u} = \mathbf{0}\)  {on $L_t^\mathcal S$}, and thus \(\mathbf{x} = \mathbf{0}\).
\end{proof}

Now we focus on lattice graphs and provide several key properties. The first one summarizes the preceding arguments.

\begin{proposition}[Property 1: Injectivity]\label{prop:T2_injective}
For every \( t \), the operators \( T_2^{(t)} \) and \( T_2'^{(t)} \) are injective.
\end{proposition}

\begin{proof}
The injectivity of \( T_2'^{(t)} \) follows directly from Lemmas~\ref{lem:linear_independence} and~\ref{Lem:Cond}. The injectivity of the operator \( T_2^{(t)} \) follows identically.
\end{proof}

Next we introduce a slightly modified graph relative to the corner subgraph $G'^{(t)}$, from which we remove one node $p$ on the boundary $L_t$. Specifically, we define 
\begin{equation}\label{eq:subgraph_Gpp}
G''^{(t)}_p := \bigl( E_p''^{(t)}, \; L_{t-1}^\mathcal{S} \setminus M_p, \; (J_{t-1}^\mathcal{S} \cup L_t \cup M_p) \setminus \{p\} \bigr),
\end{equation}
where the set $
M_p := \mathcal{N}(p) \cap \bigl( J_{t-1}^\mathcal{S} \cup L_{t-1}^\mathcal{S} \bigr)$
consists of the neighbors of \( p \) within the preceding layers, and the edge set $E_p''^{(t)}$ is given by
\[
E_p''^{(t)} := E'^{(t)} \setminus \{ pq : q \in M_p \}.
\]
Note that if the node \( p \) is adjacent to the boundary, isolated boundary nodes may arise. The boundary currents on such isolated nodes are set to  zero. Fig.~\ref{fig:Interfaces} gives a schematic illustration of the removal of the node \( p \) on the slice $t$. { Note that the ``hexagonal'' appearance in Fig. \ref{fig:Interfaces} is due to the visualization from the diagonal angle of the underlying cubic lattice. This remark applies also to Figs. \ref{fig:InterFace} and \ref{fig:Jup} below.}

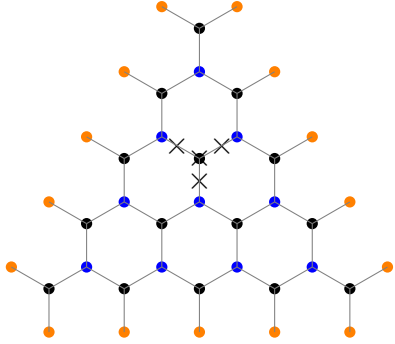
\begin{figure}[htbp]
    \centering
     
    \begin{tikzpicture}[dot/.style={circle, fill=black, inner sep=1.5pt}]
  % 定义层数 n
  \def\n{5}
  
  % 绘制格点
  \foreach \i in {1,...,\n} {
    \foreach \j in {1,...,\i} {
      % 计算点坐标：水平间距 1cm，垂直间距 \sqrt{3}/2 cm 以保持等边三角布局
      \node[dot] at (-\i/2 + \j - 0.5, -\i * 0.866) {};
      \ifnum \i = 3
        \ifnum \j = 2
           \foreach \dx/\dy in {-0.3/0.17, 0.3/0.17, 0/-0.3, 0/0} {
            \node[black] at (-\i/2 + \j - 0.5 + \dx, -\i * 0.866 + \dy) {\large\bfseries $\times$};
          }
        \fi
      \fi
    };
  \ifnum  \i > 1 
     \foreach \j in {1,...,\numexpr\i -1} {
      % 计算点坐标：水平间距 1cm，垂直间距 \sqrt{3}/2 cm 以保持等边三角布局
      \node[dot,blue] at (-\i/2 + \j - 0.5+0.5, -\i * 0.866+0.866/3) {};
      \draw[gray] (-\i/2 + \j - 0.5+0.5, -\i * 0.866+0.866/3)--(-\i/2 + \j - 0.5+0.5, -\i * 0.866+0.866 );
      \draw[gray] (-\i/2 + \j - 0.5+0.5, -\i * 0.866+0.866/3)-- (-\i/2 + \j - 0.5+1, -\i * 0.866);
      \draw[gray] (-\i/2 + \j - 0.5+0.5, -\i * 0.866+0.866/3)-- (-\i/2 + \j - 0.5, -\i * 0.866);
    };
\fi

  \node[dot,orange] at (-\i/2 + \i - 0.5+0.5, -\i * 0.866+0.866/3) {};
   \draw[gray] (-\i/2 + \i - 0.5+0.5, -\i * 0.866+0.866/3) --  (-\i/2 + \i - 0.5 , -\i * 0.866 );
 \node[dot,orange] at (-\i/2 + 1 - 0.5-0.5, -\i * 0.866+0.866/3) {};
   \draw[gray] (-\i/2 + 1 - 0.5-0.5, -\i * 0.866+0.866/3) --  (-\i/2 + 1 - 0.5 , -\i * 0.866 );

  \ifnum \i = \n
        \foreach \j in {1,...,\i} {
            \node[dot, orange] at (-\i/2 + \j - 0.5, -\i * 0.866 - 0.577) {}; % 0.577 ≈ 0.866 * 2/3
            \draw[gray]  (-\i/2 + \j - 0.5, -\i * 0.866 - 0.577)  --   (-\i/2 + \j - 0.5, -\i * 0.866 );
        }
    \fi
  }

\end{tikzpicture} 
        \caption{Node deletion in dimension $d=3$. $L_t$ $($blue$)$, $L_{t+1}$ $($black$)$ and $K_t^-$ $($orange$)$ for $t \le n+d-1$. }
    \label{fig:Interfaces}
\end{figure}

The following lemma establishes the unique recovery of the potential $\mathbf{u}$ on the reduced vertex set \( L_t^\mathcal{S} \setminus \{p\} \) from the boundary measurements on \( J_{t-1}^\mathcal{S}\setminus M_p \), i.e., the uniqueness of the interior potential recovery on the graph  \(G''^{(t)}\).

\begin{lemma}[Property 2: Uniqueness on the reduced subgraph]\label{Lem:Cond/dot}.

Let \( d \leq t \leq dn \) and fix any node \( p \in L_t \). Suppose that the conductivity \(\gamma > 0\) is specified on the subgraph \( G_p''^{(t)} \) defined in \eqref{eq:subgraph_Gpp}. Consider the boundary value problem:
\begin{equation}
   \left\{ \begin{aligned}
        \Delta_\gamma \mathbf u &= \mathbf 0 \quad \text{in } L_{t-1}^\mathcal S \setminus M_p,
       \\ \mathbf u &= \mathbf 0 \quad\text{on }J_{t-1}^\mathcal S,\\
       D_\gamma  \mathbf u &= \mathbf 0\quad \text{on }J_{t-1}^\mathcal S\setminus M_p.
    \end{aligned}\right.
\end{equation}
The only solution \(\mathbf{u}\) supported on \( L_t^\mathcal{S} \setminus \{p\} \) is the trivial solution \(\mathbf{u} = \mathbf{0}\).
\end{lemma}
\begin{proof}
Let the deleted node be located at \(p = (p_i)_{i=1}^d \). We can recover the potential \(\mathbf{u}\) in the region \(\{ (x_i)_{i=1}^d \in D : x_1 < p_1 \}\) using the argument of Lemma \ref{Lem:Cond}. By changing the inductive direction from $\mathbf e_1$ to $\mathbf e_j$ with $j >1$, \(\mathbf{u}\) can also be recovered in the region
\[
\{ (x_i)_{i=1}^d : \exists i \in \{1, \ldots, d\} \text{ such that } x_i < p_i \}.
\]
Since these elements form the set \(L_t^\mathcal{S} \setminus \{p\}\), we complete the recovery.
\end{proof}

This result yields an important relation among the vectors \(\mathbf{v}_q\) associated with the nodes in the layers \( L_t^\mathcal{S} \), given in the following corollary.
\begin{corollary}\label{Cor:w=av}
The vectors \(\mathbf{v}_q\) defined in \eqref{eqn:vp} satisfy the following properties.
\begin{enumerate}
    \item[{\rm(i)}] The collection of vectors \(\{\mathbf{v}_q|_{(L_{t-1}^\mathcal{S} \cup J_{t-1}^\mathcal S) \setminus M_p}\}_{q \in L_t^\mathcal{S} \setminus \{p\}}\) is linearly independent.
    \item[{\rm(ii)}] If \(\mathbf{w} \in \operatorname{span}\{\mathbf{v}_q|_{L_{t-1}^\mathcal{S} \cup J_{t-1}^\mathcal S }\}_{q \in L_{t}^\mathcal{S}}\) is supported on \( M_p \), then there exists \(\alpha \in \mathbb{R}\) such that
    \[
    \mathbf{w} = \alpha \mathbf{v}_p|_{L_{t-1}^\mathcal{S}\cup J_{t-1}^\mathcal S }.
    \]
\end{enumerate}
 \end{corollary}
\begin{proof}
The assertion (i) follows directly from the argument for the equivalence \textbf{(i)} \(\Leftrightarrow\) \textbf{(ii)} in Lemma~\ref{lem:linear_independence} and Lemma~\ref{Lem:Cond/dot}. To prove assertion (ii), suppose that the vector \(\mathbf{w} = \sum_{q \in L_{t}^\mathcal{S}} c_q \mathbf{v}_q|_{L_{t-1}^\mathcal{S} \cup J_{t-1}^\mathcal S}\) is supported on \( M_p \). By restricting $\mathbf{w}$ to \(( L_{t-1}^\mathcal{S}\cup J_{t-1}^\mathcal S 
)\setminus M_p \), we have
\[
\mathbf{0} = \sum_{q \in L_t^\mathcal{S} \setminus \{p\}} c_q \mathbf{v}_q|_{(L_{t-1}^\mathcal{S} \cup J_{t-1}^\mathcal S)\setminus M_p}.
\]
By part (i), it follows that all coefficients \( c_q = 0 \) for \( q \neq p \). Hence, \(\mathbf{w} = \alpha \mathbf{v}_p|_{L_{t-1}^\mathcal{S} \cup J_{t-1}^\mathcal S }\) for some \(\alpha \in \mathbb{R}\).
\end{proof}

Finally, we state an elementary observation on interface connectivity. 
\begin{lemma}[Property 3: Interface connectivity]\label{lem:interface_connectivity}
The neighborhood $\mathcal{N}(L_t \cup K_t)$ of the set \( L_t \cup K_t \) satisfies
\[
\mathcal{N}(L_t \cup K_t) \subseteq L_{t-1} \cup K_{t-1} \cup L_{t+1} \cup K_{t+1}.
\]
In particular, there are no edges connecting nodes within \( L_t \cup K_t \) itself.
\end{lemma}

\begin{remark}
{ The uniqueness analysis of this work essentially relies on the properties of the lattice graph structure listed in Proposition \ref{prop:T2_injective} and Lemmas \ref{Lem:Cond/dot} and \ref{lem:interface_connectivity}, and the analysis strategy can be extended to slightly more general graphs, e.g., ball shaped bodies cut out of hypercubic lattices. However, the analysis of the case of other graphs, e.g., hexagonal grids, likely requires substantially new techniques and is still to be investigated.}
\end{remark}

\section{Proof of Theorem \ref{thm:main}}\label{Sec:PROOF}
   
We now prove Theorem~\ref{thm:main}, i.e.,  the unique determination of \(\gamma\) from the DtN matrix \(\Lambda_\gamma\). The proof is based on mathematical induction slice by slice. {The proof proceeds in the following two steps. (i) suppose that the conductivity $\gamma$ can be recovered on the edge set $E(L_{t-1}^\mathcal S,\overline{D})$. Then we can determine the solution space $\mathcal U^{(t)}$. (ii) The solution space $\mathcal U^{(t)}$ yields a system of linear equations for $\gamma|_{E(L_{t}^\mathcal S,\overline{D})\setminus E(L_{t-1}^\mathcal S,\overline{D})}$, which is proven to be uniquely solvable in Theorem \ref{thm:uniqueness}. Thus the conductivity $\gamma$ can be uniquely recovered on the enlarged edge set $E(L_{t }^\mathcal S,\overline{D})$.}

\subsection{Extract the solution space from the DtN data $\Lambda_\gamma$}

The only available data for the recovery of the conductivity $\gamma$, aside from the graph structure, is the DtN matrix $\Lambda_\gamma$. In order to recover the conductivity $\gamma$ using the boundary measurements from vectors in the solution space $\mathcal U^{(t)}$, we first prove that the respective boundary potentials in $\ker T_1^{(t)}$ can be directly determined by $\Lambda_\gamma$. In fact, $ \ker T_1^{(t)}=\ker T^{(t)} $ , where $T^{(t)}= \Lambda_\gamma(\partial D \setminus J_t^\mathcal{S}; J_t^\mathcal{S})$ is a submatrix of $\Lambda_\gamma$:

\begin{lemma}\label{lem:Data_to_SolSpace}
For all $d-1\le t \le dn-1$, we have
\[
\ker T_1^{(t)} = \ker T^{(t)}.
\]
\end{lemma}
\begin{proof}
By the factorization identity
\[
T^{(t)} = T_2^{(t)} \circ T_1^{(t)}
\]
from Proposition~\ref{prop:factorization}, for any \(\varphi \in \ker T_1^{(t)}\), we have
\[
T^{(t)} \varphi = T_2^{(t)}(T_1^{(t)} \varphi) = T_2^{(t)}(\mathbf{0}) = \mathbf{0}.
\]
This shows the inclusion
$\ker T_1^{(t)} \subseteq \ker T^{(t)}$.
Conversely, Proposition~\ref{prop:T2_injective} guarantees that \(T_2^{(t)}\) is injective. Therefore, for any \(\varphi \in \ker T^{(t)}\), we have
\[
\mathbf{0} = T^{(t)} \varphi = T_2^{(t)}(T_1^{(t)} \varphi),
\]
which implies
$T_1^{(t)} \varphi = \mathbf{0}$.
Hence,
$\ker T^{(t)} \subseteq \ker T_1^{(t)}$.
Combining these two inclusions yields the desired equality
$\ker T_1^{(t)} = \ker T^{(t)}$.
\end{proof}

\subsection{Slice-by-slice reconstruction of conductivity}

In this part we present the proof of Theorem \ref{thm:main}. It is based on mathematical induction along the slices.
The base case begins with the layer that is closest to the origin. Note that $L_d = \{(1,1,\cdots,1)\}$ and $K_{d-1}=\{(1,\cdots,1,\underbrace{0}_{\text{the $i$th index}},1,\cdots,1):i=1,\cdots,d\}$. We begin the reconstruction of the conductivity \(\gamma\) by focusing on the subset of edges connecting the node sets \(K_{d-1}\) and \(L_d\):
\[
E(K_{d-1}, L_d) := \{\, pq \in E : p \in K_{d-1}, \ q \in L_d \,\}.
\]
\begin{lemma}\label{lem:recon-base}
 The DtN matrix $\Lambda_\gamma$ uniquely determines $\gamma$ on \(E(K_{d-1}, L_d)\).
\end{lemma}
\begin{proof}
By Lemma~\ref{lem:CHAR}, the orthogonal complement of 
\[
\mathcal{U}^{(d-1)}|_{K_{d-1}} = \ker T_1^{(d-1)}
\]
is one-dimensional. This complement is spanned by a vector that encodes precisely the conductivity values
\[
(\gamma_{pq})_{q \in K_{d-1}}, \quad \text{with } p = (1,\ldots,1) \in L_d.
\]
Thus the subspace \(\ker T_1^{(d-1)}\) uniquely determines the conductivity  \(\gamma\) on the edge set \(E(K_{d-1}, L_d)\).
\end{proof}

Next we define an edge set \(E^t\) by
\begin{equation}\label{eq:Et_partition}
E^{t} := \bigcup_{i=d-1}^{t} E_i,
\end{equation}
where for each \(i \geq d-1\), the subset
\begin{equation}\label{eq:Ei_definition}
E_i := E(K_i^-, L_{i+1}) \cup E(L_i, K_{i+1}^+) \cup E(L_i, L_{i+1})
\end{equation}
collects all edges connecting nodes in \(K_i \cup L_i\) to those in \(K_{i+1} \cup L_{i+1}\). With $\emptyset$ being the empty set, the interface graph
\begin{equation}\label{eq:Gt_definition}
G_t := \bigl(E_t, \emptyset, L_t \cup J_t \cup L_{t+1}  \bigr)
\end{equation}
thus encodes the connectivity structure relevant to the layers; see Fig.~\ref{fig:InterFace} for a schematic illustration on cubic lattices.

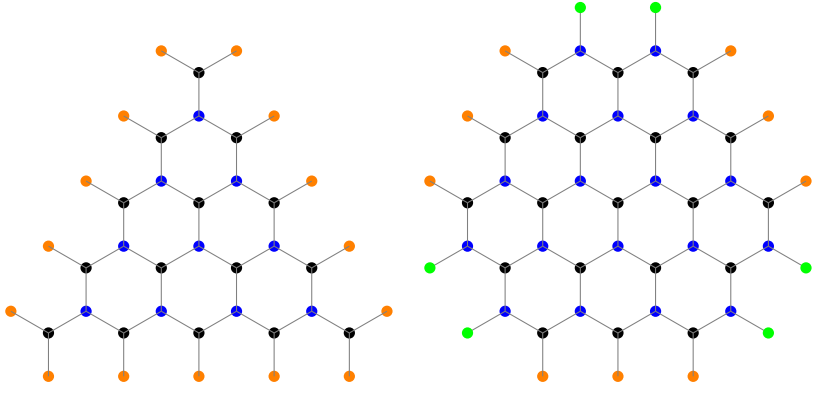
\begin{figure}
    \centering
   \begin{tabular}{cc}
    \begin{tikzpicture}[dot/.style={circle, fill=black, inner sep=1.5pt}]
  % 定义层数 n
  \def\n{5}
  
  % 绘制格点
  \foreach \i in {1,...,\n} {
    \foreach \j in {1,...,\i} {
      % 计算点坐标：水平间距 1cm，垂直间距 \sqrt{3}/2 cm 以保持等边三角布局
      \node[dot] at (-\i/2 + \j - 0.5, -\i * 0.866) {};
    };
  \ifnum  \i > 1 
     \foreach \j in {1,...,\numexpr\i -1} {
      % 计算点坐标：水平间距 1cm，垂直间距 \sqrt{3}/2 cm 以保持等边三角布局
      \node[dot,blue] at (-\i/2 + \j - 0.5+0.5, -\i * 0.866+0.866/3) {};
      \draw[gray] (-\i/2 + \j - 0.5+0.5, -\i * 0.866+0.866/3)--(-\i/2 + \j - 0.5+0.5, -\i * 0.866+0.866 );
      \draw[gray] (-\i/2 + \j - 0.5+0.5, -\i * 0.866+0.866/3)-- (-\i/2 + \j - 0.5+1, -\i * 0.866);
      \draw[gray] (-\i/2 + \j - 0.5+0.5, -\i * 0.866+0.866/3)-- (-\i/2 + \j - 0.5, -\i * 0.866);
    };
\fi

  \node[dot,orange] at (-\i/2 + \i - 0.5+0.5, -\i * 0.866+0.866/3) {};
   \draw[gray] (-\i/2 + \i - 0.5+0.5, -\i * 0.866+0.866/3) --  (-\i/2 + \i - 0.5 , -\i * 0.866 );
 \node[dot,orange] at (-\i/2 + 1 - 0.5-0.5, -\i * 0.866+0.866/3) {};
   \draw[gray] (-\i/2 + 1 - 0.5-0.5, -\i * 0.866+0.866/3) --  (-\i/2 + 1 - 0.5 , -\i * 0.866 );

  \ifnum \i = \n
        \foreach \j in {1,...,\i} {
            \node[dot, orange] at (-\i/2 + \j - 0.5, -\i * 0.866 - 0.577) {}; % 0.577 ≈ 0.866 * 2/3
            \draw[gray]  (-\i/2 + \j - 0.5, -\i * 0.866 - 0.577)  --   (-\i/2 + \j - 0.5, -\i * 0.866 );
        }
    \fi
  }

\end{tikzpicture}  &  \begin{tikzpicture}[dot/.style={circle, fill=black, inner sep=1.5pt}]
  % 定义层数 n
  \def\n{7}
  \def\cut{2}
   % 绘制格点
  \foreach \i in {1,...,\n} {
    \foreach \j in {1,...,\i} {
      % 计算点坐标：水平间距 1cm，垂直间距 \sqrt{3}/2 cm 以保持等边三角布局
      \tikzmath{
       int \ch ;
       \ch =  \n- \i+ \j;
}
\tikzmath{
       int \sh ;
       \sh =  \n+1- \j;
}
      \ifnum \i >2
      \ifnum \ch >2
      \ifnum \sh >2
      \node[dot] at (-\i/2 + \j - 0.5, -\i * 0.866) {};
      \fi
      \fi
      \fi
    };
  \ifnum  \i > 2
     \foreach \j in {1,...,\numexpr\i -1} {
      % 计算点坐标：水平间距 1cm，垂直间距 \sqrt{3}/2 cm 以保持等边三角布局
         \tikzmath{
       int \ch ;
       \ch =  \n- \i+ \j;
}
\tikzmath{
       int \sh ;
       \sh =  \n+1- \j;
}
\ifnum \ch >1
\ifnum \sh >2
      \node[dot,blue] at (-\i/2 + \j - 0.5+0.5, -\i * 0.866+0.866/3) {};
      \draw[gray] (-\i/2 + \j - 0.5+0.5, -\i * 0.866+0.866/3)--(-\i/2 + \j - 0.5+0.5, -\i * 0.866+0.866 );
      \draw[gray] (-\i/2 + \j - 0.5+0.5, -\i * 0.866+0.866/3)-- (-\i/2 + \j - 0.5+1, -\i * 0.866);
      \draw[gray] (-\i/2 + \j - 0.5+0.5, -\i * 0.866+0.866/3)-- (-\i/2 + \j - 0.5, -\i * 0.866);
     \ifnum \ch = 2
     \node[dot,green] at (-\i/2 + \j - 0.5, -\i * 0.866) {};
     \fi

    \ifnum \sh =3
     \node[dot,green] at (-\i/2 + \j - 0.5+1, -\i * 0.866) {};
    \fi

\ifnum \i =3
     \node[dot,green] at (-\i/2 + \j - 0.5+0.5, -\i * 0.866+0.866 ) {};
    \fi
     
      \fi
      \fi
    };
\fi
\tikzmath{
int \nuw;
\nuw = \n - 1;
}
\ifnum \i >2
\ifnum \i < \nuw
  \node[dot,orange] at (-\i/2 + \i - 0.5+0.5, -\i * 0.866+0.866/3) {};
   \draw[gray ] (-\i/2 + \i - 0.5+0.5, -\i * 0.866+0.866/3) --  (-\i/2 + \i - 0.5 , -\i * 0.866 );
 \node[dot,orange] at (-\i/2 + 1 - 0.5-0.5, -\i * 0.866+0.866/3) {};
   \draw[gray ] (-\i/2 + 1 - 0.5-0.5, -\i * 0.866+0.866/3) --  (-\i/2 + 1 - 0.5 , -\i * 0.866 );
\fi
\fi
  \ifnum \i = \n
        \foreach \j in {3,...,{\numexpr\i-2}} {
            \node[dot, orange] at (-\i/2 + \j - 0.5, -\i * 0.866 - 0.577) {}; % 0.577 ≈ 0.866 * 2/3
            \draw[gray ]  (-\i/2 + \j - 0.5, -\i * 0.866 - 0.577)  --   (-\i/2 + \j - 0.5, -\i * 0.866 );
        }
    \fi
  }

\end{tikzpicture}  
\end{tabular}
    \caption{ $L_t$ $($blue$)$, $L_{t+1}$ $($black$)$, $K_t^-$ $($orange$)$ and $K_{t+1}^+$ $($green$)$ for $t \le n+d-1$ $($left$)$ and $t> n+d-1$ $($right$)$ in the 3D case.}
    \label{fig:InterFace}
\end{figure}
\vspace{0.5em}

\noindent\textit{\textbf{Inductive proof of Theorem \ref{thm:main}:}} Suppose that the conductivity \(\gamma\) on the edges \(E^{t-1}\) is known (this holds true for the base case \(t=d\), cf. Lemma \ref{lem:recon-base}). We impose the boundary potential by $\varphi\in \ker T_1^{(t)}$, and record the resulting boundary current $\psi = \Lambda_\gamma \varphi$. Given the conductivity $\gamma$ on $E^{t-1}$, by Lemma~\ref{Lem:Cond}, we can determine the interior potential distribution $\mathbf u$, which is supported within $L_{t}^\mathcal S$. Then the subvector $\widetilde\gamma:=\gamma|_{E_t}$ should satisfy
$$\begin{aligned}
    \sum_{q\in\mathcal N(p)}\gamma_{pq}(u_p-u_q)&=0,\quad\text{for }p\in L_t\cup L_{t+1},\\
    \sum_{q\in\mathcal N(p)}\gamma_{pq}(u_p-u_q)&=\psi_p,\quad\text{ for }p\in J_t.
\end{aligned}$$
These relations can be written into {a system of linear equations of the only unknown $\widetilde\gamma$:} 
\begin{equation}
    \mathbf J_p^\mathbf u\cdot \widetilde\gamma =C_p, \quad \forall p\in L_t\cup L_{t+1}\cup J_t {\text{ and }\mathbf u\in \mathcal U^{(t)},}\label{eqn:JGACP}
\end{equation}
where $C_p=\psi_p$ for $p\in J_t$ and $C_p=0$ otherwise, and for each edge \(e = pq\) in $E_t$, the vector \(\mathbf{J}_p^\mathbf{u}\in\mathbb R^{E_t}\) is defined componentwise by (see Fig. \ref{fig:Jup})
\begin{equation}
    \mathbf{J}_p^\mathbf{u}(e) = 
\begin{cases}
\mathbf{u}_p - \mathbf{u}_q, & \text{if } e = pq, \\
0, & \text{otherwise}.
\end{cases}\label{eqn:JPU}
\end{equation}
The existence of a solution to the system \eqref{eqn:JGACP} is always ensured (by the assumption that it is realized by the true conductivity), and the solutions necessarily contain the true conductivity that generates the DtN matrix $\Lambda_\gamma$. Furthermore, the next theorem gives the uniqueness of the solutions and thereby gives the reconstruction of a single-layer conductivity.
\begin{theorem}
\label{thm:uniqueness}
In $\mathbb R^{E_t}$, there holds 
\begin{equation}\label{eqn:span-J-orth}
\operatorname{span}\left(\left\{ \mathbf J_p^\mathbf u : \mathbf{u} \in \mathcal{U}^{(t)},p\in L_t\cup L_{t+1}\cup J_t\right\}\right)^\perp = \{\mathbf0\}.
\end{equation}
\end{theorem}
\begin{proof}
We first investigate the space  \(\operatorname{span}\{\mathbf{J}_p^\mathbf{u} : \mathbf{u} \in \mathcal{U}^{(t)}\}\) for each node \(p\in L_t\cup L_{t+1}\cup J_t\) under the following three cases: (i) \(p \in L_{t+1}\), (ii) \(p \in L_t\) with no connection to \(K_{t+1}^+\), or \(p \in K_t^-\), and (iii) \(p \in L_t\) connected to \(q \in K_{t+1}^+\). Below we analyze the three cases separately.

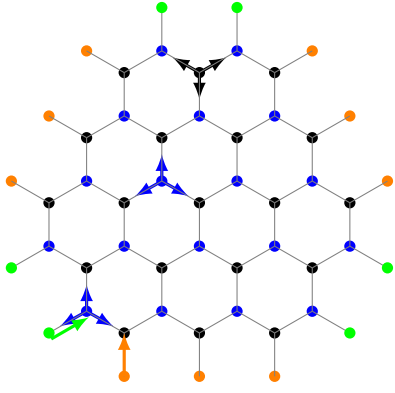
\begin{figure}[hbt!]
    \centering
    \begin{tikzpicture}[dot/.style={circle, fill=black, inner sep=1.5pt}]
  % 定义层数 n
  \def\n{7}
  \def\cut{2}
   % 绘制格点
  \foreach \i in {1,...,\n} {
    \foreach \j in {1,...,\i} {
      % 计算点坐标：水平间距 1cm，垂直间距 \sqrt{3}/2 cm 以保持等边三角布局
      \tikzmath{
       int \ch ;
       \ch =  \n- \i+ \j;
}
\tikzmath{
       int \sh ;
       \sh =  \n+1- \j;
}
      \ifnum \i >2
      \ifnum \ch >2
      \ifnum \sh >2
      \node[dot] at (-\i/2 + \j - 0.5, -\i * 0.866) {};
      \fi
      \fi
      \fi
      \ifnum \i = 3
        \ifnum \j = 2
           \foreach \dx/\dy in {-0.3/0.17, 0.3/0.17, 0/-0.3} {
            \draw[->, >=latex,very thick] (-\i/2 + \j - 0.5 , -\i * 0.866  ) -- (-\i/2 + \j - 0.5 + 1.2*\dx, -\i * 0.866 + 1.2*\dy)  ;
          }
        \fi
      \fi
    };
  \ifnum  \i > 2
     \foreach \j in {1,...,\numexpr\i -1} {
      % 计算点坐标：水平间距 1cm，垂直间距 \sqrt{3}/2 cm 以保持等边三角布局
         \tikzmath{
       int \ch ;
       \ch =  \n- \i+ \j;
}
\tikzmath{
       int \sh ;
       \sh =  \n+1- \j;
}
\ifnum \ch >1
\ifnum \sh >2
      \node[dot,blue] at (-\i/2 + \j - 0.5+0.5, -\i * 0.866+0.866/3) {};
              
          \ifnum \i = 5
        \ifnum \j = 2
           
               \foreach \dx/\dy in {-0.3/-0.17, 0.3/-0.17, 0/0.3} {
            \draw[->, >=latex,blue,very thick] (-\i/2 + \j - 0.5+0.5,-\i * 0.866+0.866/3  ) -- (-\i/2 + \j - 0.5+0.5+ 1.2*\dx, -\i * 0.866+0.866/3 + 1.2*\dy)  ;
          }
        \fi
      \fi

       \ifnum \i = 7
        \ifnum \j = 2
           
               \foreach \dx/\dy in {-0.3/-0.17, 0.3/-0.17, 0/0.3} {
            \draw[->, >=latex,blue,very thick] (-\i/2 + \j - 0.5+0.5,-\i * 0.866+0.866/3  ) -- (-\i/2 + \j - 0.5+0.5+ 1.2*\dx, -\i * 0.866+0.866/3 + 1.2*\dy)  ;
          }
          \draw[->, >=latex,green,very thick] (-\i/2 + \j - 0.5+0.03,-\i * 0.866-0.05-0.03 )  -- (-\i/2 + \j - 0.5+0.5+0.03,-\i * 0.866+0.866/3 -0.05 -0.03) ;
        \fi
      \fi
      \draw[gray] (-\i/2 + \j - 0.5+0.5, -\i * 0.866+0.866/3)--(-\i/2 + \j - 0.5+0.5, -\i * 0.866+0.866 );
      \draw[gray] (-\i/2 + \j - 0.5+0.5, -\i * 0.866+0.866/3)-- (-\i/2 + \j - 0.5+1, -\i * 0.866);
      \draw[gray] (-\i/2 + \j - 0.5+0.5, -\i * 0.866+0.866/3)-- (-\i/2 + \j - 0.5, -\i * 0.866);
     \ifnum \ch = 2
     \node[dot,green] at (-\i/2 + \j - 0.5, -\i * 0.866) {};
     \fi

    \ifnum \sh =3
     \node[dot,green] at (-\i/2 + \j - 0.5+1, -\i * 0.866) {};
    \fi

\ifnum \i =3
     \node[dot,green] at (-\i/2 + \j - 0.5+0.5, -\i * 0.866+0.866 ) {};
    \fi
     
      \fi
      \fi
    };
\fi
\tikzmath{
int \nuw;
\nuw = \n - 1;
}
\ifnum \i >2
\ifnum \i < \nuw
  \node[dot,orange] at (-\i/2 + \i - 0.5+0.5, -\i * 0.866+0.866/3) {};
   \draw[gray ] (-\i/2 + \i - 0.5+0.5, -\i * 0.866+0.866/3) --  (-\i/2 + \i - 0.5 , -\i * 0.866 );
 \node[dot,orange] at (-\i/2 + 1 - 0.5-0.5, -\i * 0.866+0.866/3) {};
   \draw[gray ] (-\i/2 + 1 - 0.5-0.5, -\i * 0.866+0.866/3) --  (-\i/2 + 1 - 0.5 , -\i * 0.866 );
\fi
\fi
  \ifnum \i = \n
        \foreach \j in {3,...,{\numexpr\i-2}} {
            \node[dot, orange] at (-\i/2 + \j - 0.5, -\i * 0.866 - 0.577) {}; % 0.577 ≈ 0.866 * 2/3
            \draw[gray ]  (-\i/2 + \j - 0.5, -\i * 0.866 - 0.577)  --   (-\i/2 + \j - 0.5, -\i * 0.866 );
        };
                    \draw[orange,->,>=latex,very thick ]  (-\i/2 + 3 - 0.5, -\i * 0.866 - 0.577)  --   (-\i/2 + 3 - 0.5, -\i * 0.866 );

    \fi
  }
\end{tikzpicture}  
\caption{A schematic illustration of $\mathbf J_p^\mathbf u$ with different $p$ in the 3D case with $L_t$ {\rm(}blue{\rm)}, $L_{t+1}$ $($black$)$, $K_t^-$ $($orange$)$ and $K_{t+1}^+$ $($green$)$.}
    \label{fig:Jup}
\end{figure}

\medskip
\noindent\textbf{Case (i): \(p \in L_{t+1}\).} Let \(M_p = (L_t \cup J_t) \cap \mathcal{N}(p)\), where the set \(\mathcal{N}(p)\) consists of the neighbors of \(p\). Consider the subspace
\[
\operatorname{span}\left\{ \mathbf{u}|_{M_p} : \mathbf{u} \in \mathcal{U}^{(t)} \right\}.
\]If \(\mathbf{w}\in\mathbb R^{M_p}\) is orthogonal to the subspace $\operatorname{span}\{ \mathbf{u}|_{M_p} : \mathbf{u} \in \mathcal{U}^{(t)}\}$, then by Lemma~\ref{lem:CHAR}, the zero extension, still denoted by  \(\mathbf{w}\in\mathbb R^{L_t^\mathcal S\cup J_{t }^\mathcal S}\), belongs to \(\operatorname{span}\{\mathbf{v}_q|_{L_t^\mathcal S\cup J_{t }^\mathcal S} : q \in L_{t+1}^\mathcal{S}\}\) and is supported in $M_p$. Corollary~\ref{Cor:w=av} implies \(\mathbf{w} = \alpha \mathbf{v}_p|_{L_t^\mathcal S\cup J_{t }^\mathcal S } \) for some \(\alpha \in \mathbb{R}\). Meanwhile, by noting $u_p=0$, we have
\begin{equation}
    \left(\operatorname{span}\{\mathbf{J}_p^\mathbf{u} : \mathbf{u} \in \mathcal{U}^{(t)}\}\right)^\perp = \operatorname{span}\{\bm{\gamma}|_{\mathcal E(p)}\} \times \mathbb{R}^{E _{t}\setminus \mathcal E(p)},\label{eq:orthogonal_complement_p0}
\end{equation}
where \(\mathcal E(p):=E(p,M_p)\) denotes the set of edges incident to $p$.

\medskip
\noindent\textbf{Case (ii): \(p \in L_t\) with no connection to \(K_{t+1}^+\), or \(p \in K_t^-\).} By Lemma~\ref{lem:interface_connectivity}, $\mathcal N(p)\cap ( L_t \cup J_t \cup L_{t+1}) = \mathcal N(p)\cap L_{t+1}$. Furthermore, we have \(\mathbf{u}|_{\mathcal{N}(p)\cap L_{t+1}} = 0\) for any $\mathbf{u}\in \mathcal{U}^{(t)}$. Let \(\mathbf{1}_{\mathcal E(p)}\) denote the indicator vector of $\mathcal{E}(p)$ (i.e., edges incident to \(p\)). Then to prove
\begin{equation}
    \operatorname{span}\{\mathbf{J}_p^\mathbf{u} : \mathbf{u} \in \mathcal{U}^{(t)}\} = \operatorname{span}\{\mathbf{1}_{\mathcal E(p)}\},\label{eq:span_case2}
\end{equation}
it suffices to show that there exists \(\mathbf{u} \in \mathcal{U}^{(t)}\) with \(\mathbf{u}_p \neq 0\). Indeed, if no such \(\mathbf{u}\) exists, then the canonical basis vector \(\mathbf{e}_p\in\mathbb R^{J_{t}^\mathcal S\cup L_{t}^\mathcal S}\) is orthogonal to \(\mathcal{U}^{(t)}\) and naturally supported on some \(M_q\) with \(q \in \mathcal{N}(p) \cap L_{t+1}\). By Corollary~\ref{Cor:w=av}, this implies \(\mathbf{e}_p = \alpha \mathbf{v}_q|_{J_{t}^\mathcal S\cup L_{t}^\mathcal S}\) for some \(\alpha \in \mathbb{R}\). {However, for any positive $\gamma$, the vector $\mathbf{v}_q|_{ J_{t}^\mathcal S\cup L_{t}^\mathcal S}$ contains more than one nonzero entries, whereas $\mathbf{e}_p$ contains only one nonzero entry, leading to a contradiction.}

\medskip
\noindent\textbf{Case (iii): \(p \in L_t\) connected to \(q \in K_{t+1}^+\).} 
By an argument analogous to Case (ii), one can readily verify that
\begin{align}\label{eq:span_case3}
&\operatorname{span}\{\mathbf{J}_p^\mathbf{u} : \mathbf{u} \in \mathcal{U}^{(t)}\} + \operatorname{span}\{\mathbf{J}_q^\mathbf{u} : \mathbf{u} \in \mathcal{U}^{(t)}\} \\
=& \underbrace{\mathbb{R}}_{\text{edge } pq} \times \operatorname{span}\{\mathbf{1}_{\mathcal E(p) \setminus \{pq\}}\} \times \underbrace{\mathbf{0}}_{\text{edges }E_t\backslash \mathcal E(p)},\nonumber
\end{align}
where the notation \(\mathbf{1}_{\mathcal E(p) \setminus \{pq\}}\) denotes the indicator vector on edges incident to \(p\) excluding the edge \(pq\). 

\noindent Now we can prove the desired assertion \eqref{eqn:span-J-orth}. Suppose that \(\mathbf{w}\) lies in the orthogonal complement, i.e.,
\begin{equation}\label{eq:w_in_orthogonal_complement}
\mathbf{w} \in \operatorname{span}\left(\left\{ \mathbf J_p^\mathbf u : \mathbf{u} \in \mathcal{U}^{(t)},p\in L_t\cup L_{t+1}\cup J_t\right\}\right)^\perp .
\end{equation}
We claim \(\mathbf{w} = \mathbf{0}\). Since for any node \(p \in L_{t+1}\), \(\mathbf{w}\) is orthogonal to each \(\mathbf{J}_p^\mathbf{u}\), its restriction $\mathbf{w}|_{\mathcal{E}(p)}$ is proportional to \(\bm{\gamma}|_{\mathcal E(p)}\) as in \eqref{eq:orthogonal_complement_p0}. Hence for any node \(p \in L_{t+1}\), there exists \(c_p \in \mathbb{R}\) such that
\begin{equation}\label{eq:w_restriction}
\mathbf{w}|_{\mathcal E(p)} = c_p \bm{\gamma}|_{\mathcal E(p)}.
\end{equation}
Next, the orthogonality condition \(\mathbf{w} \perp \mathbf{J}_p^\mathbf{u}\) for all \(p \in L_t \cup J_t\), together with the relations \eqref{eq:span_case2} and \eqref{eq:span_case3}, leads to the following system of  linear constraints
\begin{equation}\label{eq:linear_constraint_c}
\sum_{q \in \mathcal{N}(p) \cap L_{t+1}} c_q \gamma_{pq} = 0, \quad \forall p \in L_t \cup K_t^-.
\end{equation}
We extend the vector \(\mathbf{c}\) by setting
$\mathbf{c} = 0$  on  $L_t^\mathcal{S} \cup J_t^\mathcal{S}$,
so that the boundary nodes \(J_t^\mathcal{S}\) have zero potential and zero current.
Then \(\mathbf{c}\) solves the boundary value problem
\[
\left\{\begin{aligned}
\Delta_\gamma \mathbf{c} &= \mathbf{0} \quad \text{in } L_t^\mathcal{S}, \\
\mathbf{c} &= \mathbf0 \quad \text{on } J_t^\mathcal{S}, \\
D_\gamma \mathbf{c} &= \mathbf0 \quad \text{on } J_t^\mathcal{S}.
\end{aligned}\right.
\]
By Lemma~\ref{Lem:Cond}, the only solution is \(\mathbf{c} = 0\) on \(L_{t+1}^\mathcal{S}\). This and the identity \eqref{eq:w_restriction} imply that \(\mathbf{w}\) vanishes on all edges incident to nodes in \(L_{t+1}\), except possibly those incident to nodes in \(K_{t+1}^+\). However, by \eqref{eq:span_case3}, we can finally conclude that \(\mathbf{w} = \mathbf{0}\) everywhere.
This completes the proof of the theorem.
\end{proof}

\begin{theorem}
\label{thm:All}
Given the conductivity $\gamma$ on the edge set $E^{t-1}$, the data pairs
\[
\{ (\varphi, \Lambda_\gamma \varphi) : \varphi \in \ker T_1^{(t)} \}
\]
uniquely determine the conductivity \(\gamma\) on the edge set \(E_t\).
\end{theorem}

\begin{proof}
The true conductivity \(\bm\gamma:=\gamma|_{E_t}\) lies in the set 
\begin{equation*} 
\{\bm\gamma\mid\mathbf J_p^\mathbf u\cdot \bm\gamma =C_p  \text{ for all }p\in L_t\cup L_{t+1}\cup J_t,\mathbf u\in\mathcal U^{(t)}\},
\end{equation*}
with $C_p$ from \eqref{eqn:JGACP}. If $\bm\gamma_1$ and $ \bm\gamma_2$ both lie in the set, by Theorem~\ref{thm:uniqueness}, we have
\[
(\gamma_1-\gamma_2)|_{E_t}\in\operatorname{span}\left(\left\{ \mathbf J_p^\mathbf u : \mathbf{u} \in \mathcal{U}^{(t)},p\in L_t\cup L_{t+1}\cup J_t\right\}\right)^\perp  = \{\mathbf0\},
\]
which establishes the desired uniqueness.
\end{proof}

Combining Lemma~\ref{lem:Data_to_SolSpace} and Theorem~\ref{thm:All} completes the proof of Theorem~\ref{thm:main}.\hfill \(\blacksquare\)

\begin{remark}
{ For two-dimensional square lattices, Curtis and Morrow \cite{CurtisMorrow:1991} provided a sharp characterization on a subset of entries of the DtN matrix $\Lambda_\gamma$ that is necessary and sufficient for the recovery of the full DtN matrix $\Lambda_\gamma$, and using this important fact, established the unique determination of the conductivity $\gamma$ by a partial DtN matrix. For multi-dimensional hypercubic lattices, in view of the symmetry of the lattice structure, one naturally expects the existence of a subset of entries  of  the DtN matrix $\Lambda_\gamma$ for the unique recovery of $\Lambda_\gamma$. However, a precise characterization of the minimal subset of $\Lambda_\gamma$ that is necessary and sufficient for the unique recovery is still missing. Note also that the present analysis does not extend directly to the case of partial DtN matrices.}   
\end{remark}

\begin{remark}
{ The discrete Calder\'{o}n problem on cylindrical networks studied by Lam and Pylyavskyy \cite{LamPylyavskyy:2012} exhibits inherent non-uniqueness using the invariance of the DtN matrix under the so-called $Y$-$\Delta$ transformation. The key distinction between the cylindrical setting and hyper-cubic lattices is topological: cylindrical networks reside on a cylinder (or torus), whereas hyper-cubic lattices live on bounded domains in Euclidean spaces. In the Euclidean case, bounded regions possess more boundary nodes than the torus, which leads to different uniqueness results. Cylindrical networks can be reduced to networks in bounded domains in Euclidean spaces by cutting the cylinder along altitude, which gives rise to two families of boundary points along the cuts. The extension of the cylindrical setting to higher dimensions, in which networks are embedded on a higher-dimensional torus, presents intriguing possibilities that deserve further investigations.
}
\end{remark}

\section{Numerical reconstructions and discussions}
\label{Sec:Numerical}

The uniqueness proof lends itself to a slice-by-slice reconstruction algorithm, which can be viewed as a multi-dimensional analogue of the algebraic reconstruction algorithm due to Curtis and Morrow \cite{CurtisMorrow:1991} for square lattices. The detailed reconstruction procedure is listed in Algorithm \ref{alg:calderon}. When recovering the conductivity on the edge set \(E_t\), it employs a basis $\{\varphi_i\}$ of the quotient space \(\ker T^{(t)} / \ker T^{(t-1)}\). Like the Curtis-Morrow algorithm \cite{CurtisMorrow:1991}, each iteration of the algorithm involves only simple algebraic operations. Note that Algorithm~\ref{alg:calderon} starts the reconstruction process from one corner of the hypercubic lattice. However, numerically the reconstruction error grows as the iteration proceeds. The final reconstructions should combine the different solutions obtained by starting from each of the corners, and for each edge in the graph, we select the solution corresponding to the closest corner.

{ We briefly discuss the computational complexity of Algorithm \ref{alg:calderon}, with the notation $O_d(\cdot)$ indicating the involved constant depending on $d$. The main computational cost of the algorithm lies in the following three steps: line 3 (kernel evaluation), line 5 (potential recovery) and line 9 (conductivity update). First, computing the kernel \(K_t\) at line 3 involves \(O_d(n)\) evaluations, each requiring \(O_d((n^{d-1})^3)=O_d(n^{3d-3})\) operations (e.g., via singular value decomposition), leading to an overall complexity  \(O_d(n^{3d-2})\). Second, the potential recovery at line 5 involves solving for the potential with the Cauchy data (i.e., both Dirichlet and Neumann data) specified on the boundary $ J_{t-1}^\mathrm{S}$. This step
requires \(O_d(n)\) solves of (sparse) linear systems of size \(O_d(n^{d})\times O_d(n^{d})\), leading to the total complexity \(O_d(n^{1+3d})\) (by the standard direct solver without exploiting the sparse structure). Alternatively, by exploiting an upper triangular and sparse structure, this step can be reduced to \(O_d(n^{d})\) per solve, with a total complexity \(O_d(n^{1+d})\), but this procedure is numerically less stable. Third, solving for \(\bm\gamma\) in each layer at line 9 involves \(O_d(n)\) solves, each with complexity \(O_d((n^{d-1})^3)=O_d(n^{3d-3})\), resulting in a total complexity \(O_d(n^{3d-2})\). This step can be accelerated in practice by partitioning the large linear system into smaller ones, following the idea in the proof of Theorem \ref{thm:uniqueness}. In the numerical experiments below, the first step is more expensive than the other two steps. Moreover, the computation can be parallelized between different directions.}

\begin{algorithm}[hbt!]
\caption{The reconstruction algorithm for the conductivity $\bm{\gamma}$}
\label{alg:calderon}
\begin{algorithmic}[1]

\Require DtN map $\Lambda_\gamma$, grid size $n$, and dimension $d$.
\Ensure Approximation of the conductivity $\gamma$.

\State Initialize $\gamma \gets \mathbf{0}$.
\For{$t = d-1$ \textbf{to} $d \lceil \frac{n+3}{2} \rceil$}
    \State Compute the kernel
    $K_t \coloneqq \ker \Lambda_\gamma(\partial D \setminus J_t^\mathcal{S}, J_t^\mathcal{S})$, and construct a basis $\{\varphi_i\}_i$.
    \For{each $\varphi_i$ spanning $K_t/K_{t-1}$}
        \State Compute $\mathbf{u}_i \gets$ \Call{Potential}{$\gamma, \varphi_i, \Lambda_\gamma \varphi_i, t$}.
        \State Evaluate $\mathbf{J}_p^{\mathbf{u}_i}$ for all $p \in L_t \cup L_{t+1} \cup J_t$ using \eqref{eqn:JPU}.
        \State Compute $C_p$ according to \eqref{eqn:JGACP}.
    \EndFor
    \State Determine $\gamma|_{E_t}$ by solving the linear system
    \[
    \left\{ \bm{\gamma} : \mathbf{J}_p^{\mathbf{u}} \cdot \bm{\gamma} = C_p, \quad \forall p \in L_t \cup L_{t+1} \cup J_t, \quad\forall \mathbf{u} = \mathbf u_i \right\}.
    \]
\EndFor

\Function{Potential}{$\gamma, \varphi, \psi, t$}
    \State Find $\mathbf{u}$ satisfying
    \[
    \left\{\begin{aligned}
    \Delta_\gamma \mathbf{u} &= \mathbf{0}\quad  \text{in } L_{t-1}^\mathcal{S}, \\
    \mathbf{u} &= \varphi \quad \text{on } J_{t-1}^\mathcal{S}, \\
    D_\gamma \mathbf{u} &= \psi \quad \text{on } J_{t-1}^\mathcal{S}.
    \end{aligned}\right.
    \]
\EndFunction
\end{algorithmic}
\end{algorithm}

We present numerical results for cubic lattices to illustrate the feasibility of the algorithm. The relevant numerical results are shown in Figs.~\ref{fig:Recovered_Cond} and \ref{fig:Recovered_Cond2} for cubic lattices of different sizes $n$. 
The results show that the algorithm indeed can accurately recover the conductivity $\gamma$ from the DtN matrix $\Lambda_\gamma$ (when the lattice size $n$ is not large), which corroborates the proof strategy. For each fixed $n$, the pointwise error increases dramatically towards the center (i.e., the central part has the largest error), exhibiting clearly a depth-dependent resolution: the conductivity on the edges near the corners can be more accurately recovered than the ones that are far away from the corners. Moreover, as the lattice size $n$ increases from 8 to 12, the maximum error increases dramatically, from $10^{-9}$ (for $n=8$) to $10^0$ (for $n=12$), and the dominant error occurs in the central part of the cubic lattice. Nonetheless, the reconstruction accuracy near the corners is still quite satisfactory. Note that the error arises purely from the round-off error of the floating point operations when implementing the algorithm on computers. This issue is a clear manifestation of the severe ill-posed nature of the discrete Calder\'{o}n problem: the problem is expected to be exponentially ill-posed with respect to the lattice size $n$, and the phenomenon is observed also in the 2D square lattice case \cite{CurtisMorrow:1990,beretta2024discrete}. Therefore, in the presence of data noise, to obtain reasonable conductivity reconstructions, incorporating suitable regularization is indispensable. 

{ These numerical results indicate that the discrete Calder\'{o}n problem has at best logarithmic stability estimates. However, a rigorous analysis of the stability issue (with an explicit dependence on the grid size $n$ and dimensionality $d$) seems very challenging, which in the context of Algorithm \ref{alg:calderon} essentially boils down to precisely tracking the condition number of various submatrices (line 5 of the algorithm) arising from the slicing procedure. Numerically, the most unstable step lies at Step 5 of recovering the potential from the Cauchy data; Indeed, similar to the Cauchy problem for the Laplace equation, the problem is expected to be ill-conditioned. Fig. \ref{fig:kappa} presents the variation of the condition number $\kappa_t$ of the submatrices (defined as the  ratio between the largest and the smallest  singular values) arising in the potential determination step. The condition number $\kappa_t$ grows exponentially with $t$, which clearly suggests at best a logarithmic type stability estimate, but for each fixed $t$, $\kappa_t$ is observed to be nearly independent of $n$.}

\begin{figure}[hbt!]
    \centering
    \setlength{\tabcolsep}{0pt}
    \begin{tabular}{ccc}
    \includegraphics[width=0.33\linewidth]{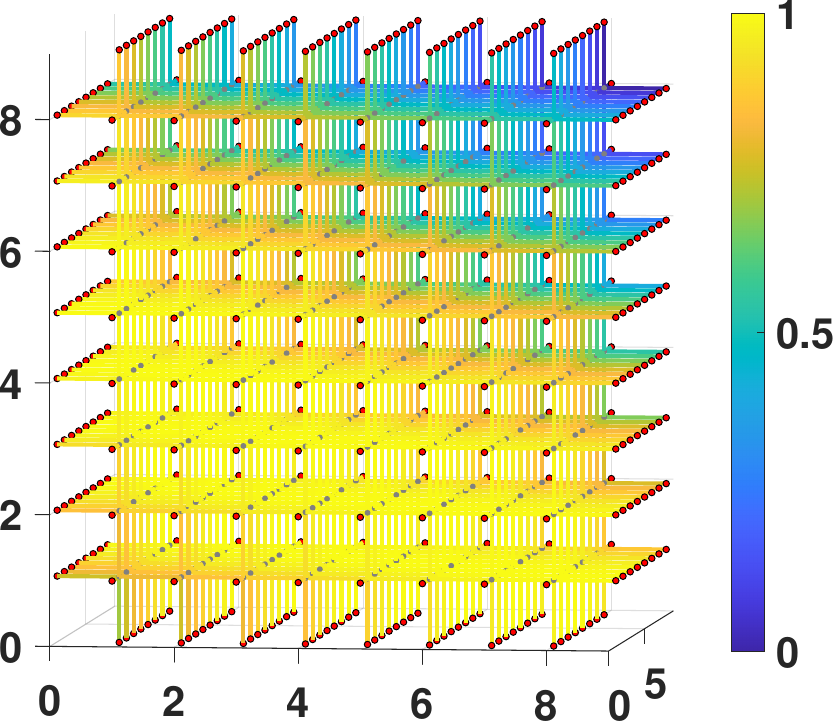}&   \includegraphics[width=0.33\linewidth]{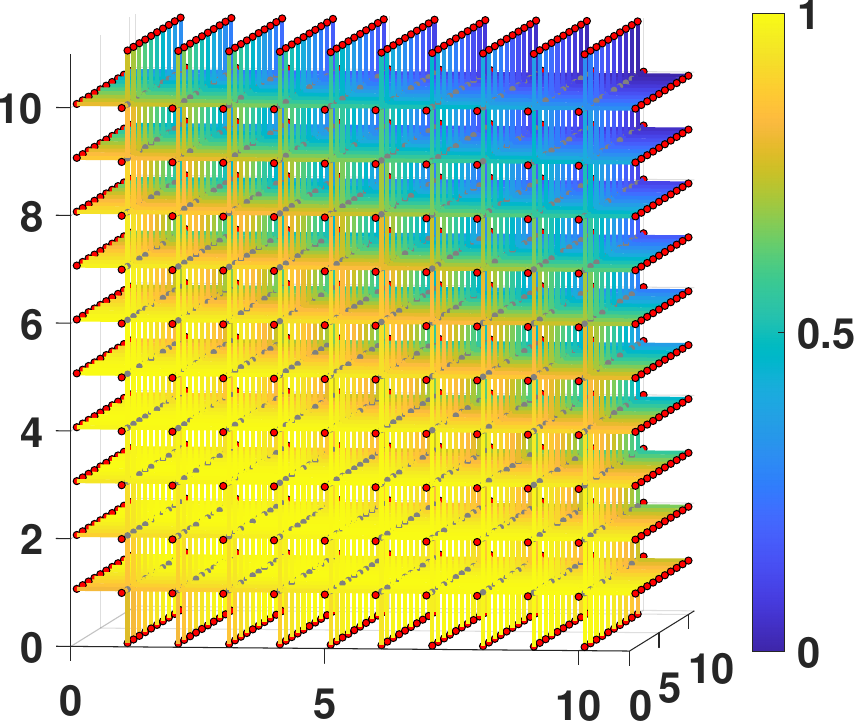}&\includegraphics[width=0.33\linewidth]{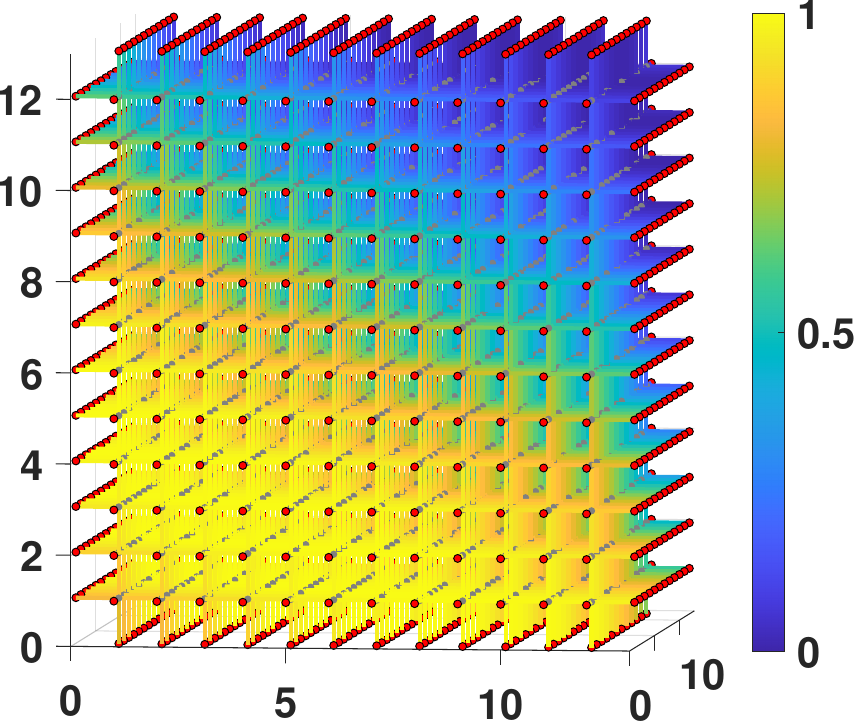}    
\\
  \includegraphics[width=0.33\linewidth]{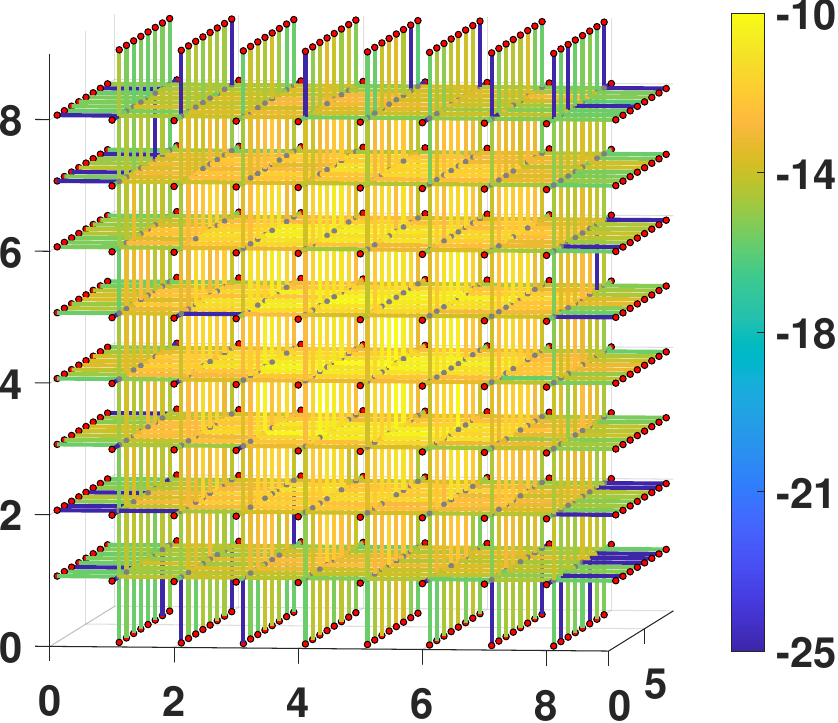}&   \includegraphics[width=0.33\linewidth]{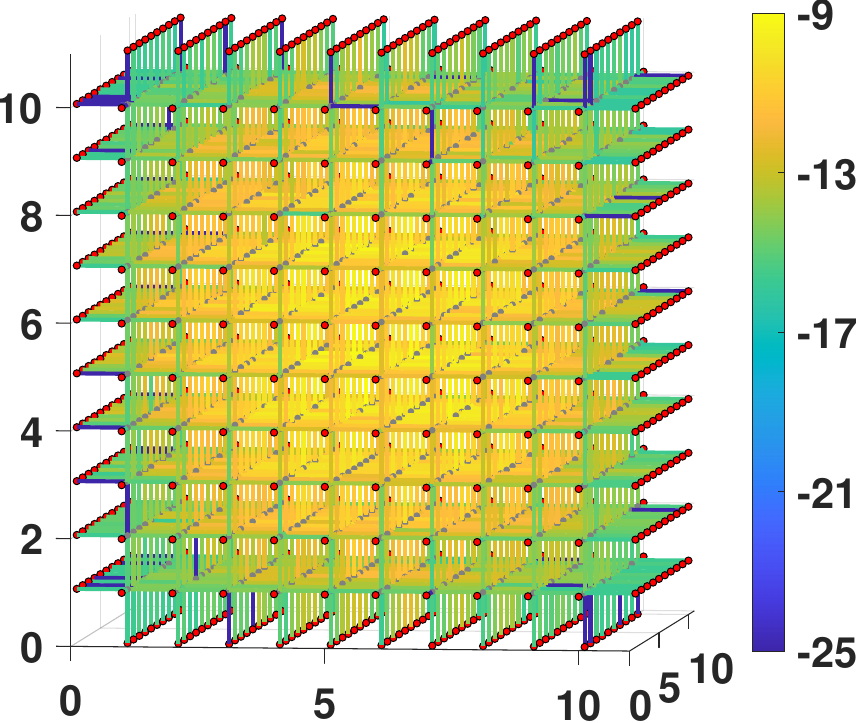}&\includegraphics[width=0.33\linewidth]{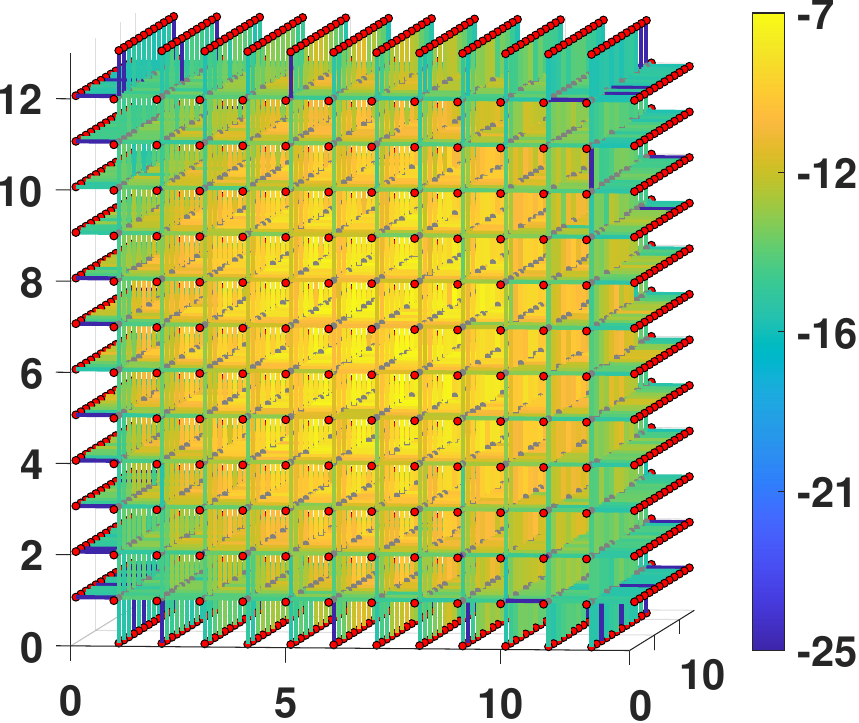}    
\\
$n={8}$&$n={10}$  &$n={12}$ \end{tabular}
    \caption{The exact conductivity $($top$)$ and the logarithm of the pointwise error $\log_{10} |\mathbf{e}_\gamma|$ $($bottom$)$ of the recovered conductivity for the three-dimensional cubic lattice.}
    \label{fig:Recovered_Cond}
\end{figure}

\begin{figure}[hbt!]
    \centering
    \setlength{\tabcolsep}{0pt}
    \begin{tabular}{ccc}
\includegraphics[width=0.33\linewidth]{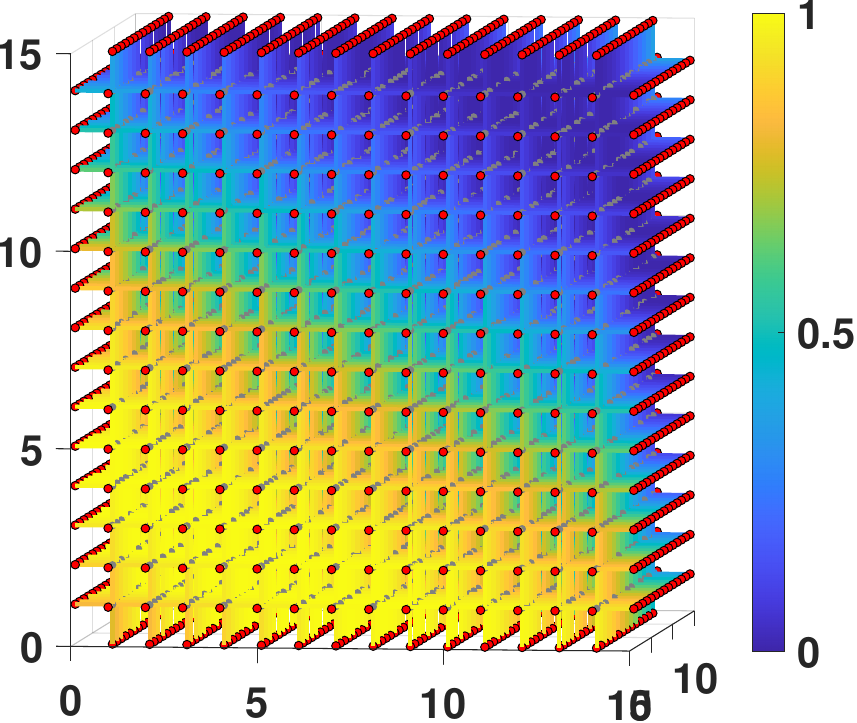}  
&   \includegraphics[width=0.33\linewidth]{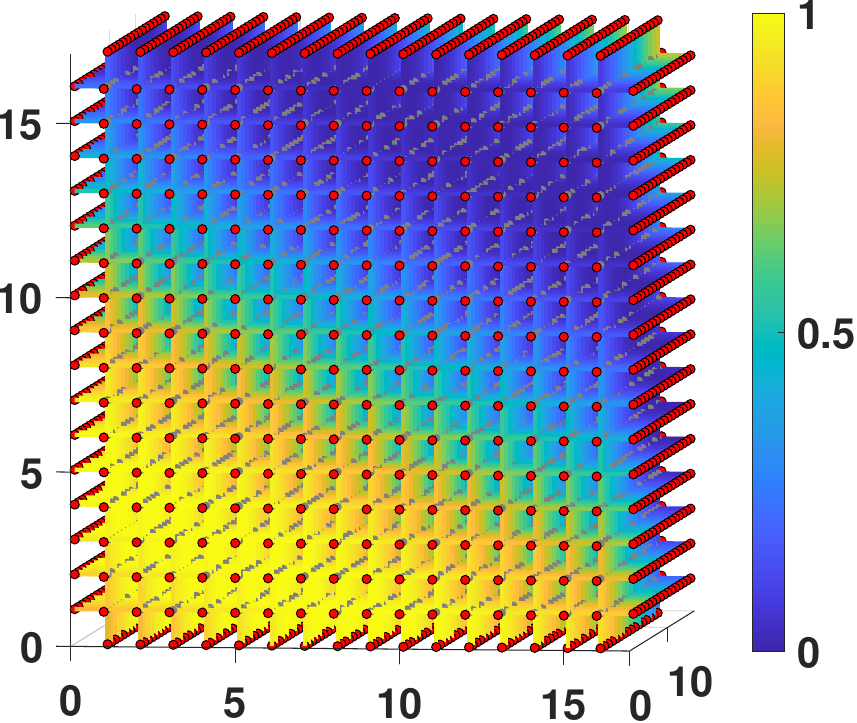}&\includegraphics[width=0.33\linewidth]{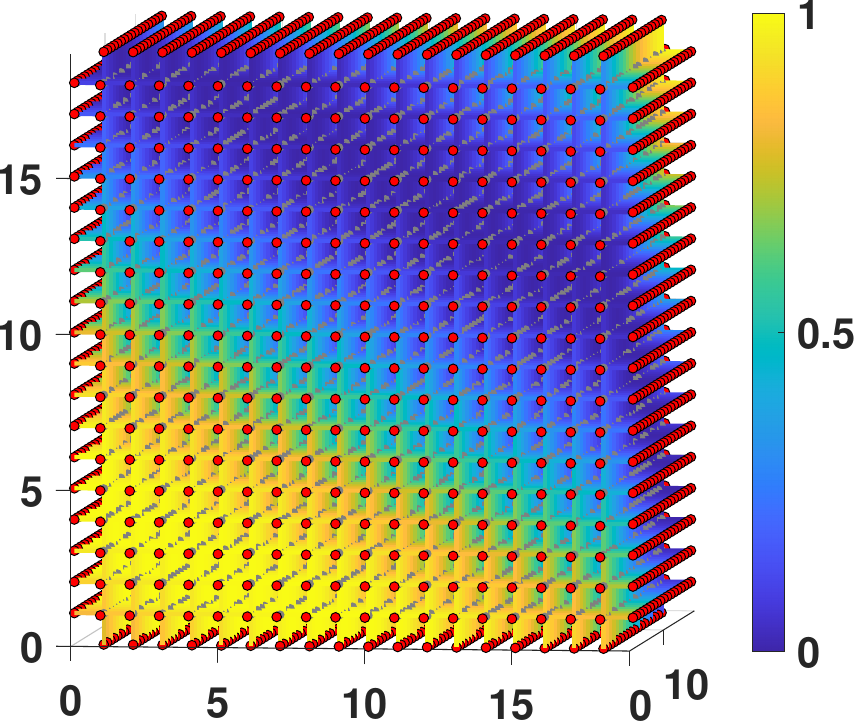}
\\
         \includegraphics[width=0.33\linewidth]{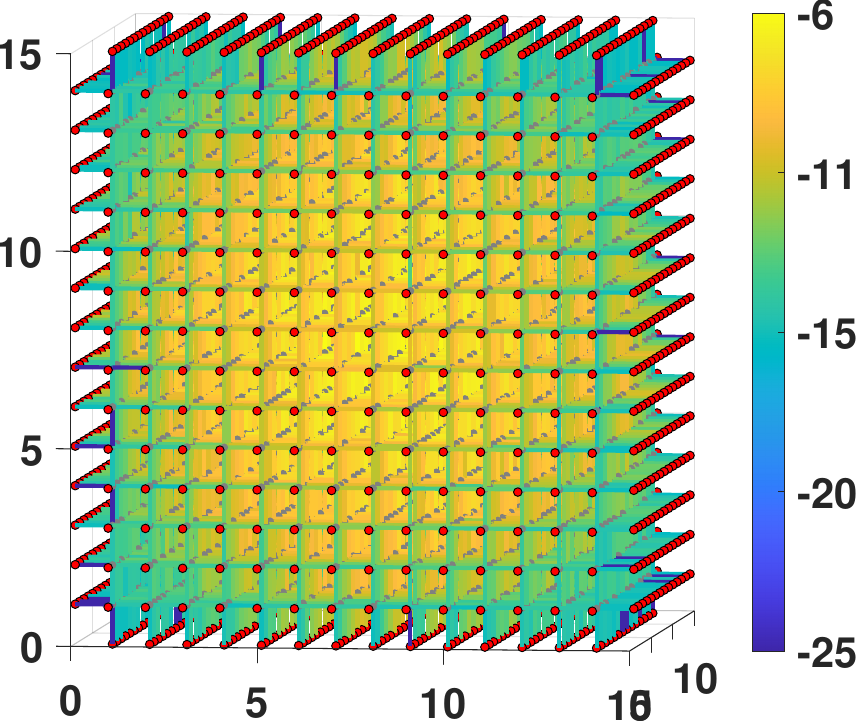}  
&   \includegraphics[width=0.33\linewidth]{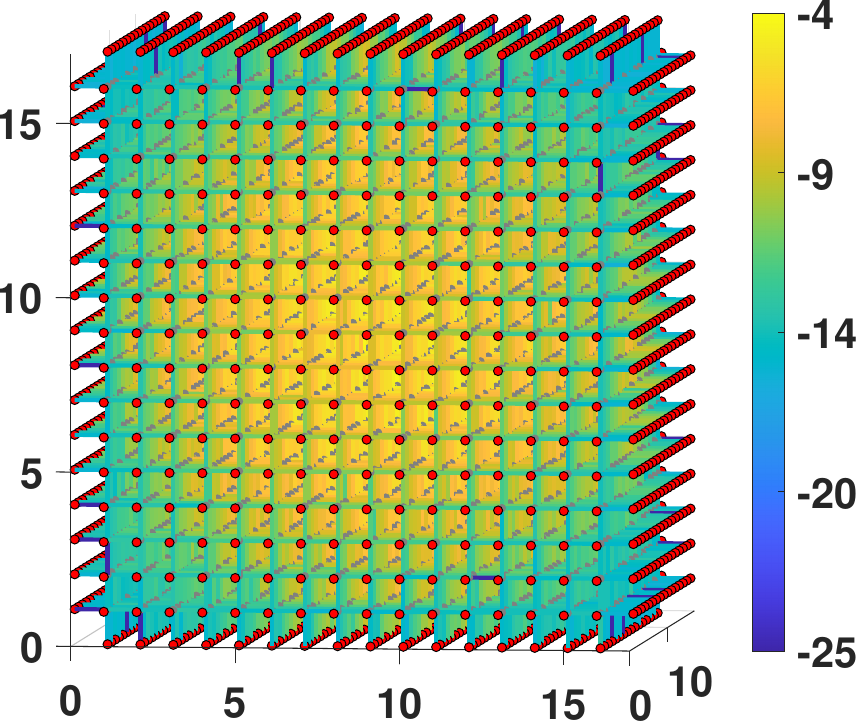}&\includegraphics[width=0.33\linewidth]{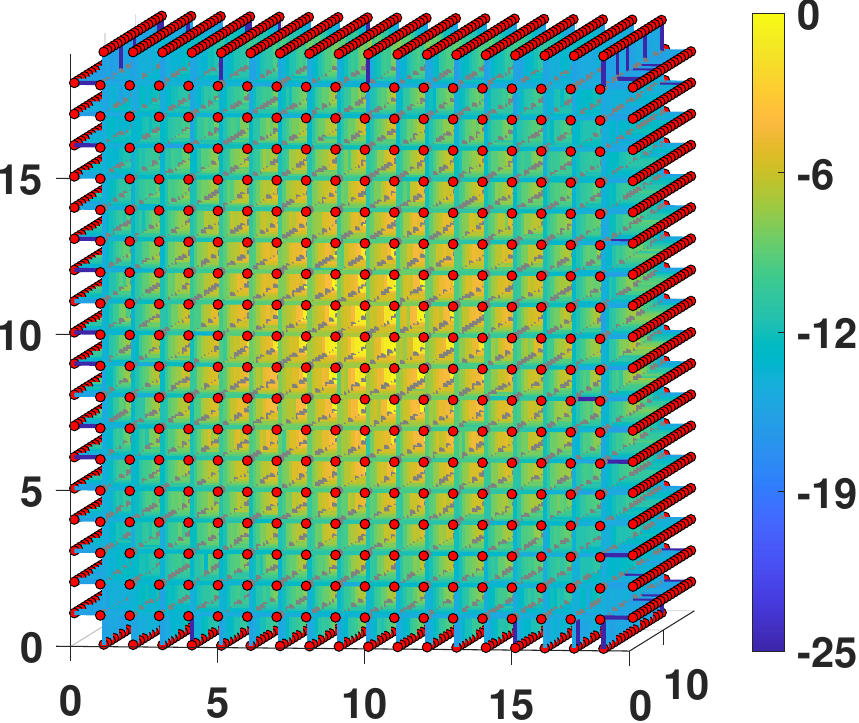}
\\$n={14}$
&$n={16}$&$n={18}$\\
    \end{tabular}
    \caption{The exact conductivity $($top$)$ and the logarithm of the pointwise error $\log_{10} |\mathbf{e}_\gamma|$ $($bottom$)$ of the recovered conductivity for the three-dimensional cubic lattice.}
    \label{fig:Recovered_Cond2}
\end{figure}

\begin{figure}[hbt!]
    \centering
    \begin{tabular}{cc}
      \includegraphics[width=0.46\linewidth]{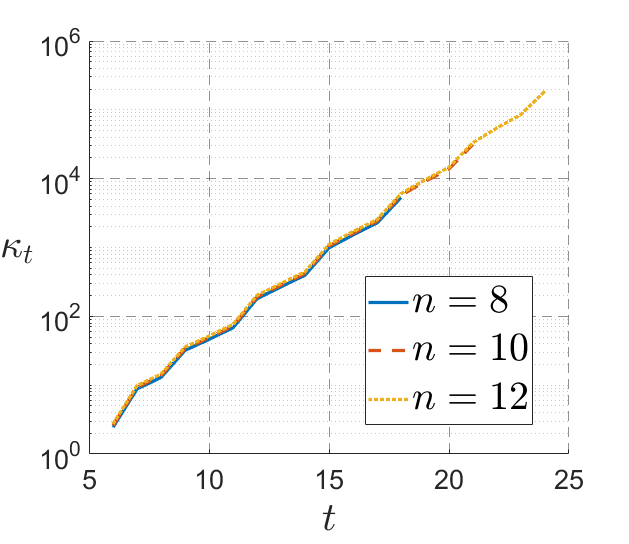}   &  \includegraphics[width=0.46\linewidth]{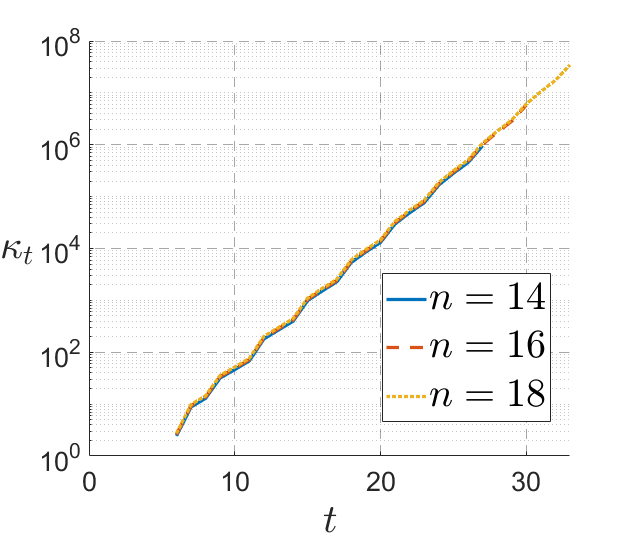}
    \end{tabular}
    \caption{{ The condition number $\kappa_t$ of the linear system involved in Step 5 of Algorithm \ref{alg:calderon}, with different $t$ varying from $6$ to  $33$. The results are for $d=3$.}}
    \label{fig:kappa}
\end{figure}

\bibliographystyle{abbrv}
\bibliography{references}

@article{CurtisMorrow:1990,
  title={Determining the resistors in a network},
  author={Curtis, Edward B and Morrow, James A},
  journal={SIAM J. Appl. Math.},
  volume={50},
  number={3},
  pages={918--930},
  year={1990},
  publisher={SIAM}
}

@article{GernandtRohleder:2022,
    author = {Gernandt, Hannes and  Rohleder, Jonathan},
    title = {A {C}alder\'{o}n type inverse problem for tree graphs},
    journal = {Linear Algebra Appl.},
    year = {2022},
    volume = {646},
    number = {},
     pages = {29--42},
}

@book{FeldmanSaloUhlmann:2025,
    author = {Feldman, Joel and Salo, Mikko and Uhlmann, Gunther},
    title = {{The {C}alder\'on Problem - An Introduction to Inverse Problems}},
    publisher = {AMS, Providence, RI},
    year = {2025}, 
}

@techreport{Oberlin:2010,
    author = {Oberlin, R},
    title = {Discrete inverse problems for Schrödinger and resistor networks},
    institution = {University of Washington, Seattle},
    year = {2010},
}

@article{Uhlmann1987,
    AUTHOR = {Sylvester, John and Uhlmann, Gunther},
     TITLE = {A global uniqueness theorem for an inverse boundary value
              problem},
   JOURNAL = {Ann. of Math. (2)},
  FJOURNAL = {Annals of Mathematics. Second Series},
    VOLUME = {125},
      YEAR = {1987},
    NUMBER = {1},
     PAGES = {153--169},
      ISSN = {0003-486X,1939-8980},
   MRCLASS = {35R30 (86A20)},
  MRNUMBER = {873380},
MRREVIEWER = {P.\ Szeptycki},
       DOI = {10.2307/1971291},
       URL = {https://doi.org/10.2307/1971291},
}

@article{beretta2024discrete,
    AUTHOR = {Beretta, Elena and Deng, Maolin and Gandolfi, Alberto and Jin,
              Bangti},
     TITLE = {The discrete inverse conductivity problem solved by the
              weights of an interpretable neural network},
   JOURNAL = {J. Comput. Phys.},
  FJOURNAL = {Journal of Computational Physics},
    VOLUME = {538},
      YEAR = {2025},
     PAGES = {114162, 24 pp.},
      ISSN = {0021-9991,1090-2716},
   MRCLASS = {35R02 (35B35 35R30 68T07 86A22 94Cxx)},
  MRNUMBER = {4918568},
       DOI = {10.1016/j.jcp.2025.114162},
       URL = {https://doi.org/10.1016/j.jcp.2025.114162},
}

@article {Uhlmann:2009,
    AUTHOR = {Uhlmann, G.},
     TITLE = {Electrical impedance tomography and {C}alder\'on's problem},
   JOURNAL = {Inverse Problems},
  FJOURNAL = {Inverse Problems. An International Journal on the Theory and
              Practice of Inverse Problems, Inverse Methods and Computerized
              Inversion of Data},
    VOLUME = {25},
      YEAR = {2009},
    NUMBER = {12},
     PAGES = {123011, 39 pp.},
      ISSN = {0266-5611,1361-6420},
   MRCLASS = {78A48 (35-02 35J25 35R30)},
  MRNUMBER = {3460047},
MRREVIEWER = {Sergey\ G.\ Pyatkov},
       DOI = {10.1088/0266-5611/25/12/123011},
       URL = {https://doi.org/10.1088/0266-5611/25/12/123011},
}

@article {CurtisMorrow:1991,
    AUTHOR = {Curtis, Edward B. and Morrow, James A.},
     TITLE = {The {D}irichlet to {N}eumann map for a resistor network},
   JOURNAL = {SIAM J. Appl. Math.},
  FJOURNAL = {SIAM Journal on Applied Mathematics},
    VOLUME = {51},
      YEAR = {1991},
    NUMBER = {4},
     PAGES = {1011--1029},
      ISSN = {0036-1399},
   MRCLASS = {31A25 (15A48 94C15)},
  MRNUMBER = {1117430},
MRREVIEWER = {Fran\c cois\ Aribaud},
       DOI = {10.1137/0151051},
       URL = {https://doi.org/10.1137/0151051},
}

@article {BoyerGarzella:2016,
    AUTHOR = {Boyer, Justin and Garzella, Jack J. and Guevara Vasquez,
              Fernando},
     TITLE = {On the solvability of the discrete conductivity and
              {S}chr\"odinger inverse problems},
   JOURNAL = {SIAM J. Appl. Math.},
  FJOURNAL = {SIAM Journal on Applied Mathematics},
    VOLUME = {76},
      YEAR = {2016},
    NUMBER = {3},
     PAGES = {1053--1075},
      ISSN = {0036-1399,1095-712X},
   MRCLASS = {35R30 (05C22 05C50 05C80 35J10 35J25)},
  MRNUMBER = {3507557},
MRREVIEWER = {Patricia\ Gaitan},
       DOI = {10.1137/15M1043479},
       URL = {https://doi.org/10.1137/15M1043479},
}

@article {BlastenIsozaki:2023,
    AUTHOR = {Bl{\aa}sten, Emilia and Isozaki, Hiroshi and Lassas, Matti and
              Lu, Jinpeng},
     TITLE = {Inverse problems for discrete heat equations and random walks
              for a class of graphs},
   JOURNAL = {SIAM J. Discrete Math.},
  FJOURNAL = {SIAM Journal on Discrete Mathematics},
    VOLUME = {37},
      YEAR = {2023},
    NUMBER = {2},
     PAGES = {831--863},
      ISSN = {0895-4801,1095-7146},
   MRCLASS = {05C50 (05C22 05C81)},
  MRNUMBER = {4598377},
MRREVIEWER = {Enno\ Pais},
       DOI = {10.1137/21M1439936},
       URL = {https://doi.org/10.1137/21M1439936},
}

@article {BlastenIsozakiLassasLu:2023,
    AUTHOR = {Bl{\aa}sten, Emilia and Isozaki, Hiroshi and Lassas, Matti and
              Lu, Jinpeng},
     TITLE = {Gelfand's inverse problem for the graph {L}aplacian},
   JOURNAL = {J. Spectr. Theory},
  FJOURNAL = {Journal of Spectral Theory},
    VOLUME = {13},
      YEAR = {2023},
    NUMBER = {1},
     PAGES = {1--45},
      ISSN = {1664-039X,1664-0403},
   MRCLASS = {35R02 (05C50 35J05 52C25)},
  MRNUMBER = {4620352},
       DOI = {10.4171/jst/455},
       URL = {https://doi.org/10.4171/jst/455},
}

@article {Ingerman:2000,
    AUTHOR = {Ingerman, David V.},
     TITLE = {Discrete and continuous {D}irichlet-to-{N}eumann maps in the
              layered case},
   JOURNAL = {SIAM J. Math. Anal.},
  FJOURNAL = {SIAM Journal on Mathematical Analysis},
    VOLUME = {31},
      YEAR = {2000},
    NUMBER = {6},
     PAGES = {1214--1234},
      ISSN = {0036-1410,1095-7154},
   MRCLASS = {39A12 (35J25 35R30)},
  MRNUMBER = {1766562},
MRREVIEWER = {Aleksander\ Denisiuk},
       DOI = {10.1137/S0036141097326581},
       URL = {https://doi.org/10.1137/S0036141097326581},
}

@article {Corbett:2025,
    AUTHOR = {Corbett, Marcus and Guevara Vasquez, Fernando and Royzman,
              Alexander and Yang, Guang},
     TITLE = {Discrete inverse problems with internal functionals},
   JOURNAL = {Inverse Problems},
  FJOURNAL = {Inverse Problems. An International Journal on the Theory and
              Practice of Inverse Problems, Inverse Methods and Computerized
              Inversion of Data},
    VOLUME = {41},
      YEAR = {2025},
    NUMBER = {4},
     PAGES = {045004, 24 pp.},
      ISSN = {0266-5611,1361-6420},
   MRCLASS = {35R30 (35J57 35R02 65K10 65M32 94C05)},
  MRNUMBER = {4882173},
       DOI = {10.1088/1361-6420/adbd69},
       URL = {https://doi.org/10.1088/1361-6420/adbd69},
}

@incollection {BorceaDruskinMamonov:2013,
    AUTHOR = {Borcea, Liliana and Druskin, Vladimir and Guevara Vasquez,
              Fernando and Mamonov, Alexander V.},
     TITLE = {Resistor network approaches to electrical impedance
              tomography},
 BOOKTITLE = {{Inverse Problems and Applications: Inside Out. {II}}},
     PAGES = {55--118},
 PUBLISHER = {Cambridge Univ. Press, Cambridge},
      YEAR = {2013},
      ISBN = {978-1-107-03201-9},
   MRCLASS = {65N21 (35J25 35R30 65N08)},
  MRNUMBER = {3098656},
MRREVIEWER = {Ruben\ D.\ Spies},
}

@incollection {Calderon:1980,
    AUTHOR = {Calder\'on, Alberto-P.},
     TITLE = {On an inverse boundary value problem},
 BOOKTITLE = {Seminar on {N}umerical {A}nalysis and its {A}pplications to
              {C}ontinuum {P}hysics ({R}io de {J}aneiro, 1980)},
     PAGES = {65--73},
 PUBLISHER = {Soc. Brasil. Mat., Rio de Janeiro},
      YEAR = {1980},
   MRCLASS = {35R30 (35K60)},
  MRNUMBER = {590275},
MRREVIEWER = {J.\ R.\ Cannon},
}

@article {BorceaMamonov:2017,
    AUTHOR = {Borcea, Liliana and Guevara Vasquez, Fernando and Mamonov,
              Alexander V.},
     TITLE = {A discrete {L}iouville identity for numerical reconstruction
              of {S}chr\"odinger potentials},
   JOURNAL = {Inverse Probl. Imaging},
  FJOURNAL = {Inverse Problems and Imaging},
    VOLUME = {11},
      YEAR = {2017},
    NUMBER = {4},
     PAGES = {623--641},
      ISSN = {1930-8337,1930-8345},
   MRCLASS = {35R30 (05C22 35J25)},
  MRNUMBER = {3668344},
MRREVIEWER = {Enno\ Pais},
       DOI = {10.3934/ipi.2017029},
       URL = {https://doi.org/10.3934/ipi.2017029},
}

@article {Morioka:2011,
    AUTHOR = {Morioka, Hisashi},
     TITLE = {Inverse boundary value problems for discrete {S}chr\"odinger
              operators on the multi-dimensional square lattice},
   JOURNAL = {Inverse Probl. Imaging},
  FJOURNAL = {Inverse Problems and Imaging},
    VOLUME = {5},
      YEAR = {2011},
    NUMBER = {3},
     PAGES = {715--730},
      ISSN = {1930-8337,1930-8345},
   MRCLASS = {35R30 (35J10 39A12 65N21)},
  MRNUMBER = {2825735},
       DOI = {10.3934/ipi.2011.5.715},
       URL = {https://doi.org/10.3934/ipi.2011.5.715},
}

@article {HorvathMarko:2016,
    AUTHOR = {Horv\'ath, Mikl\'os and Mark\'o, Zolt\'an},
     TITLE = {Discrete inverse problems for the {S}chr\"odinger operator on
              the multi-dimensional square lattice with partial {C}auchy
              data},
   JOURNAL = {Inverse Problems},
  FJOURNAL = {Inverse Problems. An International Journal on the Theory and
              Practice of Inverse Problems, Inverse Methods and Computerized
              Inversion of Data},
    VOLUME = {32},
      YEAR = {2016},
    NUMBER = {5},
     PAGES = {055006, 9},
      ISSN = {0266-5611,1361-6420},
   MRCLASS = {35R30 (35J25 65N21)},
  MRNUMBER = {3488514},
MRREVIEWER = {Enno\ Pais},
       DOI = {10.1088/0266-5611/32/5/055006},
       URL = {https://doi.org/10.1088/0266-5611/32/5/055006},
}

@article {BorceaDruskinMamonov:2010a,
    AUTHOR = {Borcea, L. and Druskin, V. and Mamonov, A. V.},
     TITLE = {Circular resistor networks for electrical impedance tomography
              with partial boundary measurements},
   JOURNAL = {Inverse Problems},
  FJOURNAL = {Inverse Problems. An International Journal on the Theory and
              Practice of Inverse Problems, Inverse Methods and Computerized
              Inversion of Data},
    VOLUME = {26},
      YEAR = {2010},
    NUMBER = {4},
     PAGES = {045010, 30},
      ISSN = {0266-5611,1361-6420},
   MRCLASS = {65N21 (35J25 35R30)},
  MRNUMBER = {2608623},
MRREVIEWER = {Dinh Nho H\`ao},
       DOI = {10.1088/0266-5611/26/4/045010},
       URL = {https://doi.org/10.1088/0266-5611/26/4/045010},
}

@article {BorceaDruskinMamonov:2010,
    AUTHOR = {Borcea, L. and Druskin, V. and Mamonov, A. V. and Guevara
              Vasquez, F.},
     TITLE = {Pyramidal resistor networks for electrical impedance
              tomography with partial boundary measurements},
   JOURNAL = {Inverse Problems},
  FJOURNAL = {Inverse Problems. An International Journal on the Theory and
              Practice of Inverse Problems, Inverse Methods and Computerized
              Inversion of Data},
    VOLUME = {26},
      YEAR = {2010},
    NUMBER = {10},
     PAGES = {105009, 36},
      ISSN = {0266-5611,1361-6420},
   MRCLASS = {65N21 (78A70 92C55)},
  MRNUMBER = {2719770},
       DOI = {10.1088/0266-5611/26/10/105009},
       URL = {https://doi.org/10.1088/0266-5611/26/10/105009},
}

@article {CurtisIngermanMorrow:1998,
    AUTHOR = {Curtis, E. B. and Ingerman, D. and Morrow, J. A.},
     TITLE = {Circular planar graphs and resistor networks},
   JOURNAL = {Linear Algebra Appl.},
  FJOURNAL = {Linear Algebra and its Applications},
    VOLUME = {283},
      YEAR = {1998},
    NUMBER = {1-3},
     PAGES = {115--150},
      ISSN = {0024-3795,1873-1856},
   MRCLASS = {05C38 (05C40 05C50 94C05)},
  MRNUMBER = {1657214},
MRREVIEWER = {W.-K.\ Chen},
       DOI = {10.1016/S0024-3795(98)10087-3},
       URL = {https://doi.org/10.1016/S0024-3795(98)10087-3},
}

@article {deVerdiGitler:1996,
    AUTHOR = {Colin de Verdi\`ere, Yves and Gitler, Isidoro and Vertigan,
              Dirk},
     TITLE = {R\'eseaux \'electriques planaires. {II}},
   JOURNAL = {Comment. Math. Helv.},
  FJOURNAL = {Commentarii Mathematici Helvetici},
    VOLUME = {71},
      YEAR = {1996},
    NUMBER = {1},
     PAGES = {144--167},
      ISSN = {0010-2571,1420-8946},
   MRCLASS = {05C10 (05C90 94C99)},
  MRNUMBER = {1371682},
MRREVIEWER = {Martin\ \v Skoviera},
       DOI = {10.1007/BF02566413},
       URL = {https://doi.org/10.1007/BF02566413},
}

@article {BorceaDruskin:2008,
    AUTHOR = {Borcea, Liliana and Druskin, Vladimir and Guevara Vasquez,
              Fernando},
     TITLE = {Electrical impedance tomography with resistor networks},
   JOURNAL = {Inverse Problems},
  FJOURNAL = {Inverse Problems. An International Journal on the Theory and
              Practice of Inverse Problems, Inverse Methods and Computerized
              Inversion of Data},
    VOLUME = {24},
      YEAR = {2008},
    NUMBER = {3},
     PAGES = {035013, 31},
      ISSN = {0266-5611,1361-6420},
   MRCLASS = {78A70 (35J25 35Q60 65N21)},
  MRNUMBER = {2421967},
       DOI = {10.1088/0266-5611/24/3/035013},
       URL = {https://doi.org/10.1088/0266-5611/24/3/035013},
}

@article {CurtisMorrow:1994,
    AUTHOR = {Curtis, E. and Mooers, E. and Morrow, J.},
     TITLE = {Finding the conductors in circular networks from boundary
              measurements},
   JOURNAL = {RAIRO Mod\'el. Math. Anal. Num\'er.},
  FJOURNAL = {RAIRO Mod\'elisation Math\'ematique et Analyse Num\'erique},
    VOLUME = {28},
      YEAR = {1994},
    NUMBER = {7},
     PAGES = {781--814},
      ISSN = {0764-583X},
   MRCLASS = {65N99 (94C05)},
  MRNUMBER = {1309415},
       DOI = {10.1051/m2an/1994280707811},
       URL = {https://doi.org/10.1051/m2an/1994280707811},
}

@article {ChungGilbert:2017,
    AUTHOR = {Chung, Francis J. and Gilbert, Anna C. and Hoskins, Jeremy G.
              and Schotland, John C.},
     TITLE = {Optical tomography on graphs},
   JOURNAL = {Inverse Problems},
  FJOURNAL = {Inverse Problems. An International Journal on the Theory and
              Practice of Inverse Problems, Inverse Methods and Computerized
              Inversion of Data},
    VOLUME = {33},
      YEAR = {2017},
    NUMBER = {5},
     PAGES = {055016, 21},
      ISSN = {0266-5611,1361-6420},
   MRCLASS = {05C90 (65J22)},
  MRNUMBER = {3634446},
       DOI = {10.1088/1361-6420/aa66d1},
       URL = {https://doi.org/10.1088/1361-6420/aa66d1},
}

@article {LamPylyavskyy:2012,
    AUTHOR = {Lam, Thomas and Pylyavskyy, Pavlo},
     TITLE = {Inverse problem in cylindrical electrical networks},
   JOURNAL = {SIAM J. Appl. Math.},
  FJOURNAL = {SIAM Journal on Applied Mathematics},
    VOLUME = {72},
      YEAR = {2012},
    NUMBER = {3},
     PAGES = {767--788},
      ISSN = {0036-1399,1095-712X},
   MRCLASS = {94C05 (05C60)},
  MRNUMBER = {2968749},
       DOI = {10.1137/110846476},
       URL = {https://doi.org/10.1137/110846476},
}

@incollection {Arauz:2016,
    AUTHOR = {Ara\'uz, C. and Carmona, \'A. and Encinas, A. M. and Mitjana,
              M.},
     TITLE = {Recovering the conductances on grids: a theoretical
              justification},
 BOOKTITLE = {{A Panorama of Mathematics: Pure and Applied}},
    SERIES = {Contemp. Math.},
    VOLUME = {658},
     PAGES = {149--166},
 PUBLISHER = {Amer. Math. Soc., Providence, RI},
      YEAR = {2016},
      ISBN = {978-1-4704-1668-3},
   MRCLASS = {35R30 (35J10 35N25 35R02 78A46 78A48)},
  MRNUMBER = {3475279},
MRREVIEWER = {Hideo\ Nakazawa},
       DOI = {10.1090/conm/658/13122},
       URL = {https://doi.org/10.1090/conm/658/13122},
}

@article {ChungBerenstein:2005,
    AUTHOR = {Chung, Soon-Yeong and Berenstein, Carlos A.},
     TITLE = {{$\omega$}-harmonic functions and inverse conductivity
              problems on networks},
   JOURNAL = {SIAM J. Appl. Math.},
  FJOURNAL = {SIAM Journal on Applied Mathematics},
    VOLUME = {65},
      YEAR = {2005},
    NUMBER = {4},
     PAGES = {1200--1226},
      ISSN = {0036-1399,1095-712X},
   MRCLASS = {35R30 (94C12)},
  MRNUMBER = {2147325},
MRREVIEWER = {Alexander\ Yurjevich\ Chebotarev},
       DOI = {10.1137/S0036139903432743},
       URL = {https://doi.org/10.1137/S0036139903432743},
}

@article {Borcea,
    AUTHOR = {Borcea, Liliana},
     TITLE = {Electrical impedance tomography},
   JOURNAL = {Inverse Problems},
  FJOURNAL = {Inverse Problems. An International Journal on the Theory and
              Practice of Inverse Problems, Inverse Methods and Computerized
              Inversion of Data},
    VOLUME = {18},
      YEAR = {2002},
    NUMBER = {6},
     PAGES = {R99--R136},
      ISSN = {0266-5611,1361-6420},
   MRCLASS = {92C55 (35Q60 78A55)},
  MRNUMBER = {1955896},
       DOI = {10.1088/0266-5611/18/6/201},
       URL = {https://doi.org/10.1088/0266-5611/18/6/201},
}

\end{document}